\documentclass[11pt]{article}

\usepackage{times}
\usepackage{amsmath}
\usepackage{algorithm}
\usepackage{amsthm}
\usepackage{graphicx,epsfig}
\usepackage{caption}
\usepackage[noend]{algpseudocode}
\theoremstyle{plain} \numberwithin{equation}{section}
\newtheorem{theorem}{Theorem}[section]
\newtheorem{corollary}[theorem]{Corollary}

\newtheorem{lemma}[theorem]{Lemma}

\theoremstyle{plain}
\newtheorem{invariant}[theorem]{Invariant}
\newtheorem{property}[theorem]{Property}

\newtheorem{observation}[theorem]{Observation}

\newtheorem{claim}{Claim}[section]
\newcommand{\N}{\mathcal{N}}

\newcommand{\C}{\mathcal{C}}
\renewcommand{\P}{\mathcal{P}}
\newcommand{\eps}{\epsilon}

\newcommand{\E}{\mathcal{E}}

\newcommand{\polylog}{\mathrm{polylog}}

\advance \topmargin by -\headheight
\advance \topmargin by -\headsep
\evensidemargin \oddsidemargin
\usepackage{xcolor}
\usepackage{xspace}
\usepackage[shortlabels]{enumitem}

\usepackage[
style=alphabetic,
backref=true,
doi=false,
url=false,
maxcitenames=3,
mincitenames=3,
maxnames=10,
minbibnames=10,
backend=bibtex8,
sortlocale=en_US
]{biblatex}

\usepackage[ocgcolorlinks]{hyperref} 
\usepackage{cleveref}
\allowdisplaybreaks

\renewcommand{\paragraph}[1]{\medskip\noindent{\bf #1}\xspace}

\colorlet{DarkRed}{red!50!black}
\colorlet{DarkGreen}{green!50!black}
\colorlet{DarkBlue}{blue!50!black}

\hypersetup{
	linkcolor = DarkRed,
	citecolor = DarkGreen,
	urlcolor = DarkBlue,
	bookmarks = true,
	bookmarksnumbered = true,
	linktocpage = true
}

\renewcommand{\polylog}{\operatorname{polylog}}

\begin{document}

\title{Dynamic Algorithms for Graph Coloring}

\author{Sayan Bhattacharya\thanks{Corresponding author. University of Warwick, UK. Email: {\tt S.Bhattacharya@warwick.ac.uk}}  \and Deeparnab Chakrabarty\thanks{Dartmouth College, USA. Email: {\tt deeparnab.chakrabarty@dartmouth.edu}} \and Monika Henzinger\thanks{University of Vienna, Austria. Email: {\tt monika.henzinger@univie.ac.at}} \and Danupon Nanongkai\thanks{KTH, Sweden. Email: {\tt danupon@gmail.com}}}

\date{}

\newcounter{list}

\newcommand{\D}{\Delta}

\begin{titlepage}
	\maketitle
	\pagenumbering{roman}
	\begin{abstract}
	We design fast dynamic algorithms for  proper vertex and edge colorings in a  graph undergoing edge insertions and deletions. In the static setting, there are simple linear time algorithms  for $(\Delta+1)$-	vertex coloring and $(2\Delta-1)$-edge coloring in a graph with maximum degree $\Delta$. It is natural  to ask if we can efficiently maintain such colorings in the dynamic setting as well. We get the following three results.
(1)  We present a {\em randomized} algorithm which maintains a $(\Delta+1)$-vertex coloring with $O(\log \Delta)$ expected amortized update time.
(2) We present a {\em deterministic} algorithm which maintains a $(1+o(1))\Delta$-vertex coloring with $O(\polylog \Delta)$ amortized update time.
(3) We present a simple, deterministic  algorithm  which maintains a $(2\Delta-1)$-edge coloring with $O(\log \Delta)$ {\em worst-case} update time. This improves the recent $O(\Delta)$-edge coloring algorithm with $\tilde{O}(\sqrt{\Delta})$ worst-case update time~\cite{BarenboimM17}.
	\end{abstract}

	\newpage
	\setcounter{tocdepth}{2}
	\tableofcontents
\end{titlepage}

\newpage
\pagenumbering{arabic}

\section{Introduction}
\label{sec:intro}

Graph coloring is a fundamental problem with many applications in computer science. A proper $c$-vertex coloring of a graph assigns a color in $\{1,\ldots,c\}$ to every node, in such a way that the endpoints of every edge get different colors. The chromatic number of the graph is the smallest $c$ for which a proper $c$-vertex coloring exists. Unfortunately, from a computational perspective, approximating the chromatic number is rather futile: for any constant $\epsilon>0$, there is no polynomial time algorithm that approximates the chromatic number within a factor of $n^{1-\epsilon}$  in an $n$-vertex graph, assuming  $P\neq NP$  \cite{FeigeK98,Zuckerman07} (see  \cite{KhotP06} for a stronger bound).
On the positive side, we know that the chromatic number is at most $\Delta+1$ where $\Delta$ is the maximum degree of the graph. There is a simple linear time algorithm to find a $(\Delta+1)$-coloring: pick any uncolored vertex $v$, scan the colors used by its neighbors, and assign to $v$ a color not assigned to any of its neighbors.
Since the number of neighbors is at most $\Delta$, by pigeon hole principle such a color must exist.

In this paper, we consider the graph coloring problem in the dynamic setting, where the edges of a graph are being inserted or deleted over time and we want to maintain a proper coloring after every update. The objective is to use as few colors as possible while keeping the update time\footnote{There are two notions of update time: {\em amortized update time} -- an algorithm has amortized update time of $\alpha$ if for any $t$, after $t$ insertions or deletions the total update time is $\leq \alpha t$, and {\em worst case update time} --  an algorithm has worst case update time of $\alpha$ if every update time $\leq \alpha$. As typical for amortized update time guarantees, we assume that the input graph is empty initially.} small. Specifically, our main goal is to investigate whether a $(\Delta+1)$-vertex coloring can be maintained with small update time. Note that the greedy algorithm described in the previous paragraph can easily be modified to give a worst-case update time of $O(\Delta)$: if an edge $(u,v)$ is inserted between two nodes $u$ and $v$ of same color, then scan the at most $\Delta$ neighbors of $v$ to find a free color. A natural question is whether we can get  an algorithm with significantly lower update time. We answer this question in the affirmative. 
\begin{itemize}
	\item  We design and analyse a {\em randomized} algorithm which maintains a $(\Delta+1)$-vertex coloring with $O(\log \Delta)$ expected amortized update time.\footnote{As typically done for randomized dynamic algorithms, we assume that the adversary who fixes the sequence of edge insertions and deletions is {\em oblivious} to the randomness in our algorithm.}
\end{itemize}
\noindent It is not difficult to see that if we had $(1+\epsilon)\Delta$ colors, then there would be a simple randomized algorithm with $O(1/\epsilon)$-expected amortized update time (see Section~\ref{sec:idea:randomized} for details). What is challenging in our result above is to maintain a $\Delta+1$ coloring with small update time. In contrast, if randomization is not allowed, then even maintaining a $O(\Delta)$-coloring with $o(\Delta)$-update time seems non-trivial. Our second result is on deterministic vertex coloring algorithms: although we do not achieve a $\Delta+1$ coloring, we come close.
\begin{itemize}	
	\item We design and analyse a {\em deterministic} algorithm which maintains a $(\Delta+o(\Delta))$-vertex coloring with $O(\polylog \Delta)$ amortized update time. 
\end{itemize}
Note that in a dynamic graph the maximum degree $\Delta$ can change over time. Our results hold with the changing $\Delta$ as well. However, for ease of explaining our main ideas we restrict most of the paper to the setting where $\Delta$ is a static upper bound known to the algorithm. In Section~\ref{sec:general} we point out the changes needed to make our algorithms work for the changing-$\Delta$ case.

Our final result is on maintaining an {\em edge coloring} in a dynamic graph with maximum degree $\Delta$. A proper edge coloring is a coloring of edges such that no two adjacent edges have the same color. 
\vspace{-1ex}
\begin{itemize}
	\item
We design and analyze a simple, deterministic $(2\Delta-1)$-edge coloring algorithm with $O(\log \Delta)$ {\em worst-case} update time.
\end{itemize}
\vspace{-1ex}
This significantly improves upon the recent $O(\Delta)$-edge coloring algorithm of Barenboim and Maimon~\cite{BarenboimM17} which needs $\tilde{O}(\sqrt{\Delta})$-worst-case update time.

\paragraph{Perspective:}
An important aspect of $(\Delta+1)$-vertex coloring is the following {\em local-fixability property}: 
Consider a graph problem $\P$ where we need to assign a {\em state} (e.g. color) to each node. We say that a constraint is {\em local to a node $v$} if it is defined on the states of $v$ and its neighbors. We say that a problem $\P$ is locally-fixable iff it has the following three properties. (i) There is a  local constraint on every node. (ii) A solution $S$ to $\P$  is {\em feasible} iff $S$ satisfies the local constraint at every node. (iii) If the local constraint $C_v$ at a node $v$ is unsatisfied, then we can change only the state of $v$ to satisfy $C_v$ without creating any new unsatisfied constraints at other nodes. 
For example, $(\Delta+1)$-vertex coloring is locally-fixable as we can define a constraint local to $v$ to be satisfied if and only if $v$'s color is different from all its neighbors, and if not, then we can always find a recoloring of $v$ to satisfy  its local constraint and introduce no other constraint violations. On the other hand, the following problems do not seem to be locally-fixable: globally optimum coloring,  the best approximation algorithm for coloring~\cite{Halldorsson93},  $\Delta$-coloring (which always exists by Brook's theorem, unless the graph is a clique or an odd cycle), and  $(\Delta+1)$-edge coloring (which always exists by Vizing's theorem).

Observe that if we start with a feasible solution for a locally-fixable problem $\P$, then after inserting or deleting an edge $(u, v)$ we need to change only the states of $u$ and $v$ to obtain a new feasible solution. For instance in the case of $(\Delta+1)$-vertex coloring, we need to recolor only the nodes incident to the inserted edge. Thus, the number of changes is guaranteed to be small, and the main challenge is to {\em search} for these changes in an efficient manner without having to scan the whole neighborhood. In contrast, for  non-locally-fixable problems, the main challenge seems to be analyzing how many nodes or edges we need to recolor (even with an inefficient algorithm) to keep the coloring proper. A question in this spirit has been recently studied in~\cite{BarbaCKLRRV-WADS17}.

It can be shown that the $(2\Delta-1)$-edge coloring problem is also locally-fixable (see Appendix~\ref{sec:appendix}).  Given our current results on $(\Delta+1)$-vertex coloring and $(2\Delta-1)$-edge coloring, it is inviting to ask whether there is some deeper connections that exist in designing dynamic algorithms for these problems. In particular, are there reductions possible among these problems? Or can we find a {\em complete} locally-fixable problem? It is also very interesting to understand the power of randomization for these problems.

 
Indeed, in the {\em distributed computing} literature, there is deep and extensive work on and beyond the locally-fixable problems above. (In fact, it can be shown that any locally-fixable problem is in the {\sf SLOCAL} complexity class studied in distributed computing~\cite{GhaffariKM17}; see Appendix~\ref{sec:appendix}.)
Coincidentally, just like our findings in this paper, there is still a big gap between deterministic and randomized distributed algorithms for $(\Delta+1)$-vertex coloring. 
For further details we refer the to the excellent monograph by Barenboim and Elkin \cite{BarenboimE13} (see \cite{GhaffariKM17,FischerGK17}, and references therein, for more recent results).

Finally, we also note that the dynamic problems we have focused on are {\em search problems}; i.e., their solutions always exist, and the hard part is to find and maintain them. This posts a new challenge when it comes to proving conditional lower bounds for dynamic algorithms for these locally-fixable problems: while a large body of work has been devoted to decision problems \cite{Patrascu10,AbboudW14,HenzingerKNS15,KopelowitzPP16,Dahlgaard16}, it seems non-trivial to adapt existing techniques to search problems.

\paragraph{Other Related Work.}
Dynamic graph coloring is a natural problem and there have been many works on this~\cite{DGOP07,OuerfelliB2011,SallinenIPGRP16,HardyLT17}. Most of these papers, however, have proposed heuristics and described experimental results. The only two theoretical papers  that we are aware of are \cite{BarenboimM17} and \cite{BarbaCKLRRV-WADS17}, and they are already mentioned above.

\paragraph{Organisation of the rest of the paper.} In Section~\ref{sec:ideas} we give the high level ideas behind our vertex coloring result. 
In particular, Section~\ref{sec:idea:randomized} contains the main ideas of the randomized algorithm, whereas Section~\ref{sec:randomized:coloring} contains the full details. Similarly, Section~\ref{sec:idea:deterministic} contains the main ideas of the deterministic algorithm, whereas Section~\ref{sec:det:coloring} contains the full details. Section~\ref{sec:edge-coloring} contains the edge-coloring result.   {\em We emphasize that Sections~\ref{sec:randomized:coloring},~\ref{sec:det:coloring},~\ref{sec:edge-coloring} are completely self contained, and they can be read independently of each other.}
As mentioned earlier, in Sections~\ref{sec:randomized:coloring},~\ref{sec:det:coloring},~\ref{sec:edge-coloring} we assume that the parameter $\Delta$ is known and that the maximum degree never exceeds $\Delta$. We do so solely for the better exposition of the main ideas. Our algorithms easily modify to give results where $\Delta$ is the current maximum degree. See Section~\ref{sec:general} for the details.

\section{Our Techniques for Dynamic Vertex Coloring}
\label{sec:ideas}

\subsection{An overview of our randomized algorithm.}
\label{sec:idea:randomized}
We present a high level overview of our randomized dynamic algorithm for $(\Delta+1)$-vertex  coloring that has $O(\log \Delta)$ expected amortized update time. The full details can be found in Section~\ref{sec:randomized:coloring}. 
We start with a couple of warm-ups before sketching the main idea.

\medskip
\noindent {\bf Warmup I: Maintaining a $2\Delta$-coloring in $O(1)$ expected amortized update time.} 
\label{sec:randomized:warmup:I}   We first observe that maintaining a $2\Delta$ coloring is easy using randomization against an oblivious adversary -- we need only $O(1)$ expected amortized time. 
The algorithm is this. Let $\C$ be the palette of $2\Delta$ colors.
Each vertex  $v$ stores the last time $\tau_v$ at which it was recolored. If an edge gets deleted  or if an edge gets inserted between two vertices   of different colors, then we do nothing. Next, consider the scenario where an edge  gets inserted at time $\tau$ between two vertices   $u$ and $v$ of same color. Without any loss of generality, suppose that $\tau_v > \tau_u$, i.e., the vertex $v$ was recolored last. In this event, we scan all the neighbors of $v$ and store the colors used by them in a set $S$, and select a random color from $\C\setminus S$. Since $|\C| = 2\Delta$, we have $|\C\setminus S| \geq \Delta$ as well.
Clearly this leads to a proper coloring since the new color of $v$, by design, is not the current colors of any of $v$'s neighbors.

The time taken to compute the set $S$ can be as high as $O(\Delta)$ since $v$ can have $\Delta$ neighbors. Now, let us analyze the probability that the insertion of the edge $(u,v)$ at time $\tau$ leads to a conflict. Suppose that at time $\tau_v$, just before $v$ recolored itself, the color of $u$ was $c$. The insertion at time $\tau$ creates a conflict only if $v$ chose the color $c$ at time $\tau_v$. However the probability of this event is at most $1/\Delta$, since $v$ had at least $\Delta$ choices to choose its color from at time $\tau_v$. Therefore the expected time spent on the addition of edge $(u,v)$ is  $O(\Delta) \cdot (1/\Delta) = O(1)$.

In the analysis described above, we have crucially used the fact that the insertion of the edge $(u,v)$ at time $\tau$ is oblivious to the random choice made  while recoloring the vertex $v$ at time $\tau_v$. It should also be clear that the constant $2$ is not sacrosanct and a $(1+\eps)\Delta$ coloring can be obtained in $O(1/\eps)$-expected amortized time. However this fails to give a $\Delta+1$ or even $\Delta + c$ coloring in $o(\Delta)$ time for any constant $c$.

\medskip
\noindent {\bf Warmup II: A simple algorithm for $(\Delta+1)$ coloring that is difficult to analyze.} 
\label{sec:randomized:warmup:II} 
 In the previous algorithm, while  recoloring a vertex we made sure that it never assumed the color of any of its neighbors. We say that a color $c$ is {\em blank} for a vertex $v$ iff no neighbor of $v$ gets the color $c$. Since we have $\Delta+1$ colors, every vertex has at least one blank color. However, if there is only one blank color to choose from, then an adversarial sequence of updates may force the algorithm to spend $\Omega(\Delta)$ time after every edge insertion. A polar-opposite idea would be to randomly recolor a vertex without considering the colors of its neighbors. This has the problem that a recoloring may lead to one or more neighbors of the vertex being unhappy (i.e., having the same color as $v$), and then there is a cascading effect which is hard to control. 

We take the middle ground: Define a color $c$ to be {\em unique} for a vertex $v$ if  it is assigned to {\em exactly one} neighbor of $v$. Thus, if $v$ is recolored using a unique color then the cascading effect of unhappy vertices   doesn't explode. Specifically, after recoloring $v$ we only need to consider recoloring $v$'s unique neighbor, and so on and so forth. 
Why is this idea useful? This is because although the number of {\em blank} colors available to a vertex (i.e., the colors which none of its neighbors are using) can be as small as $1$, the number of blank+unique colors is always at least $\Delta/2$. This holds since any color which is neither blank nor unique accounts for at least two neighbors of $v$, whereas $v$ has at most $\Delta$ neighbors.

The above observation suggests the following natural algorithm. When we need to recolor a vertex $v$, we first scan all its neighbors to identify the set $S$ of all unique {\em and} blank colors for $v$, and then we pick a new color $c$ for $v$ uniformly at random from this set $S$. By definition of the set $S$, at most one neighbor $y$ of $x$ will have the same color $c$. If such a neighbor $y$ exists, then we recolor $y$ recursively using the same algorithm. We now state three important properties of this scheme. (1) While recoloring a vertex $x$ we have to make {\em at most one} recursive call. (2) It takes $O(\Delta)$ time to recolor a vertex $x$, ignoring the potential recursive call to its neighbor. (3) When we recolor a vertex $x$, we pick its new color uniformly at random from a set of size $\Omega(\Delta)$. Note that the properties (2) and (3) served as the main tools in establishing the $O(1)$ bound on the expected amortized update time as discussed in the previous algorithm. For property (1), if we manage to upper bound the length of the {\em chain of recursive calls} that might result after the insertion of an edge in the input graph between two vertices   of same color, then we will get an upper bound on the overall update time of our algorithm. This, however, is not trivial. In fact, the reader will observe that it is not necessary to have $\Delta+1$ colors in order to ensure the above three properties. They hold even with $\Delta$ colors. Indeed, in that case the algorithm described above might never terminate. We conclude that  another idea is required to achieve $O(\log \Delta)$ update time. This turns out to be the concept of a hierarchical partition of the set of vertices   of a graph.
We describe this and present an overview of our final algorithm below.

\medskip
\noindent {\bf An overview of the final algorithm.}
\label{sec:randomized:final}
Fix a large constant $\beta > 1$, and suppose that we can partition the vertex-set of the input graph $G = (V, E)$ into $L = \log_{\beta} \Delta$ {\em levels} $1, \ldots, L$ with the following property. 
\begin{property}
	\label{prop:hierarchy}
	Consider any vertex $v$  at a level $1 \leq \ell(v) \leq L$. Then the vertex $v$ has at most $O(\beta^{\ell(v)})$ neighbors in levels $\{1, \ldots, \ell(v)\}$, and at least $\Omega(\beta^{\ell(v)-5})$ neighbors in levels $\{1, \ldots, \ell(v)-1\}$. 
\end{property}

It is not clear at first glance that there even exists such a partition of the vertex-set: Given a {\em static} graph $G = (V, E)$, there seems to be no obvious way to assign a level $\ell(v) \in \{1, \ldots, L\}$ to each vertex $v \in V$ satisfying Property~\ref{prop:hierarchy}. One of our main technical contributions is to present an algorithm that maintains a hierarchical partition satisfying Property~\ref{prop:hierarchy} in a dynamic graph. Initially, when the input graph $G = (V, E)$ has an empty edge-set, we place every vertex at level $1$. This trivially satisfies Property~\ref{prop:hierarchy}. Subsequently, after every insertion or deletion of an edge in $G$, our algorithm updates the hierarchical partition in a way which ensures that Property~\ref{prop:hierarchy} continues to remain satisfied. This algorithm is deterministic, and using an intricate charging argument we show that it has an amortized update time of $O(\log \Delta)$. This also gives a constructive proof of  the existence of a hierarchical partition that satisfies Property~\ref{prop:hierarchy} in any given graph.

We now explain how this hierarchical partition, in conjunction with the ideas from Warmup II, leads to an efficient randomized $\Delta+1$ vertex  coloring algorithm. In this algorithm, we require that a vertex $u$  keeps all its neighbors $v$ at levels $\ell(v) \leq \ell(u)$ informed about its own color $\chi(u)$. This  requirement allows a vertex $x$ to maintain: (1) the set $\C^+_x$ of colors assigned to its neighbors $y$ with $\ell(y) \geq \ell(x)$, and (2) the set $\C_x = \C \setminus \C_x^+$  of remaining colors. We say that a color $c \in \C_x$ is {\em blank} for $x$ iff no neighbor $y$ of $x$ with $\ell(y) < \ell(x)$ has the same color $c$. On the other hand, we say that a color $c \in \C_x$ is {\em unique} for $x$ iff exactly one neighbor $y$ of $x$ with $\ell(y) < \ell(x)$ has the same color $c$. Note the crucial change in the definition of a unique color from Warmup II. Now,  for a color $c$ to be unique for $x$ it is not enough that $x$ has exactly one neighbor with the same color; in addition, this neighbor has to lie at a level {\em strictly} below the level of $x$. 
Using the property of the hierarchical partition that $x$ has $\Omega(\beta^{\ell(x)-5})$ neighbors in levels $\{1, \ldots, \ell(x)-1 \}$ and an argument similar to one used in Warmup II, we can show that there are a large number of colors that are either blank or unique for $x$.  

\begin{claim}
	\label{cl:advanced}
	For every vertex $x$, there are at least $1 + \Omega(\beta^{\ell(x)-5})$ colors that are either blank or unique.
\end{claim}

We now implement the same template as in Warmup II. When a vertex $x$ needs to be recolored, it picks its new color uniformly at random from the set of its blank + unique colors. This can cause some other vertex $y$ to be unhappy, but such a vertex $y$ lies at a level strictly lower than $\ell(x)$. As there are $O(\log \Delta)$ levels, this bounds the depth of any recursive call: At level $1$, we just use a blank color. Further, whenever we recolor $x$, the  time it needs to inform all its neighbors $y$ with $\ell(y) \leq \ell(x)$ is  bounded by $O(\beta^{\ell(x)})$ (by the property of the hierarchical partition).
Since each recursive call is done on a vertex at a strictly lower level, the  total time spent on all the recursive calls can also be bounded by $O(\beta^{\ell(x)})$  due to a geometric sum. Finally, by Claim~\ref{cl:advanced}, each time $x$ picks a random color it does so from a palette of size $\Omega(\beta^{\ell(x)-5})$. If the order of edge insertions and deletions is oblivious to this randomness, then the probability that an edge insertion is going to be problematic is  $O(1/\beta^{\ell(x)-5})$, which gives an expected amortized time bound of $O(1/\beta^{\ell(x)-5}) \times O(\beta^{\ell(x)}) = O(\beta^5) = O(1)$.

\subsection{An overview of our deterministic algorithm.}
\label{sec:idea:deterministic}

We present a high level overview of  our deterministic dynamic algorithm for $(1+o(1)) \Delta$-vertex  coloring that has an amortized update time of $O(\polylog \Delta)$. The full details are in Section~\ref{sec:det:coloring}.  As in Section~\ref{sec:idea:randomized}, we start with a warmup before sketching the main idea.

\medskip
\noindent {\bf Warmup: Maintaining a $4\Delta$ coloring in $O(\sqrt{\Delta})$  amortized update time.}
Let $\C$ be the palette of $4\Delta$ colors. We partition the set $\C$ into $2\sqrt{\Delta}$ equally sized subsets: $\C_1, \ldots, \C_{2\sqrt{\Delta}}$ each having $2\sqrt{\Delta}$ colors. Colors in $\C_t$ are said to be of {\em type} $t$ and we let $t(v)$ denote the type of the color assigned to a node $v$.
Furthermore, we let $d_t(v)$ denote the number of neighbors of $v$ that are assigned a type $t$ color.
We refer to the  neighbors $u$ of $v$ with $t(u) = t$ as {\em type $t$ neighbors of $v$}. For every node $v$, we let $\Gamma(v)$ denote the set of neighbors $u$ of $v$ with $t(u) = t(v)$. Every node $v$ maintains the set $\Gamma(v)$ in a doubly linked list. Note that if the node $v$ gets a color from $\C_t$, then we have $d_t(v) = |\Gamma(v)|$.
We maintain a proper coloring  with the following extra property: If a node $v$ is of type $t$, then it has at most $2\sqrt{\Delta}-1$ type $t$ neighbors. 
\begin{property}
	\label{prop:deterministic}
	If any node $v$ is assigned a color from $\C_t$, then  we have $d_{t}(v) < 2 \sqrt{\Delta}$.  
\end{property}

Initially, the input graph $G = (V, E)$ is empty, every vertex is colored arbitrarily, and the above property holds. Note that the  deletion of an edge from $G$ does not lead to a violation of the above property, nor does it make the existing coloring invalid. We now discuss what we do when an edge $(u, v)$ gets inserted into $G$, by considering three possible cases. 

\smallskip
\noindent {\em Case 1: $t(u) \neq t(v)$.} There is nothing to be done since $u$ and $v$ have different types of colors.
\smallskip

\noindent {\em Case 2: $t(u) = t(v)=t$, but both $d_{t}(u)$ and $d_{t}(v) < 2 \sqrt{\Delta}-1$ after the insertion of the edge $(u, v)$.} 
The colors assigned to the vertices   $u$ and $v$ are of the same  type. In this event, we first set $\Gamma(u) = \Gamma(u) \cup \{v\}$ and $\Gamma(v) = \Gamma(v) \cup \{u\}$. There is nothing further to do if $u$ and $v$ don't have the same color since the property continues to hold.
If they have the same color $c$, then we pick an arbitrary endpoint $u$ and find a type $t$ color $c'\neq c$ that is not assigned to any of the neighbors of $u$ in the set $\Gamma(u)$.
This is possible since $|\Gamma(u)| = d_t(u) < 2\sqrt{\Delta}-1$ and there are $2\sqrt{\Delta}$ colors of each type.  We then change the color of $u$ to $c'$. These operations take $O(|\C_t| + |\Gamma(u)|) = O(\sqrt{\Delta})$ time.
\smallskip

\noindent {\em Case 3: $t(u) = t(v) = t$ and $d_{t}(u) = 2\sqrt{\Delta}-1$ after the insertion of the edge $(u, v)$.} 
 Here, after the addition of the edge $(u,v)$, the vertex 
$u$ violates Property~\ref{prop:deterministic}.  We run the following subroutine RECOLOR($u$):
\begin{itemize}
\item Since $u$ has at most $\Delta$ neighbors and there are  $2\sqrt{\Delta}$ types, there must exist a type $t'$ with $d_{t'}(u) \leq \sqrt{\Delta}/2$. Such a type $t'$ can be found by doing a linear scan of all the neighbors of $u$, and this takes $O(\Delta)$ time since $u$ has at most $\Delta$ neighbors.

From the set $\C_{t'}$ we choose a color $c$ that is not assigned to any of the neighbors of $u$: Such a color must exist since $|\C_{t'}| = 2\sqrt{\Delta}  > d_{t'}(u)$. Next, we update the set $\Gamma(u)$ as follows: We delete from $\Gamma(u)$ every neighbor $x$  of $u$ with $t(x) = t$, and insert  into $\Gamma(u)$ every neighbor $x$ of $u$ with $t(x) = t'$. We similarly update the set $\Gamma_x$ for every neighbor $x$ of $u$ with $t(x) \in \{t, t'\}$.  It takes $O(d_t(u) + d_{t'}(u)) = O(\sqrt{\Delta})$ time to implement this step.

Accordingly, the total time spent on this call to the RECOLOR(.) subroutine is $O(\Delta) + O(\sqrt{\Delta}) = O(\Delta)$. However,  Property~\ref{prop:deterministic} may now be violated for one or more neighbors $u'$ of $u$. If this is the case, then we recursively call  RECOLOR($u'$) and keep doing so until all the vertices   satisfy  Property~\ref{prop:deterministic}. In the end, we have a proper coloring with all the vertices   satisfying  Property~\ref{prop:deterministic}.
\end{itemize}

A priori it may not be clear that the above procedure even terminates. However, we now argue that the amortized time spent in all the calls to the RECOLOR subroutine  is $O(\sqrt{\Delta})$ (and in particular the chain of recursive calls to the subroutine terminates). To do so we introduce a potential $\Phi := \sum_{v \in V} d_{t(v)}(v)$, which sums over all vertices   the number of its neighbors which are of the same type as itself. Note that when an edge $(u,v)$ is inserted or deleted the potential can increase  by at most $2$. However, during a call to RECOLOR($u$) the potential $\Phi$ drops by at least $3\sqrt{\Delta}$.
This is because $u$ moves from a color of type $t_{old}$ to a color of type $t_{new}$ where $d_{t_{old}}(u) = 2\sqrt{\Delta}$ and $d_{t_{new}}(u) \leq \sqrt{\Delta}/2$; this leads to a drop of $~ 1.5\sqrt{\Delta}$ and we get the same amount of drop when considering $u$'s neighbors. Therefore, during $T$ edge insertions or deletions starting from an empty graph, we can have at most $O(T/\sqrt{\Delta})$ calls to the RECOLOR subroutine. Since each such call takes $O(\Delta)$ time, we get the claimed $O(\sqrt{\Delta})$ amortized update time.

\paragraph{Getting $O(\polylog \Delta)$ amortized update time.} One way to interpret the previous algorithm is as follows. Think of each color $c \in \C$ as an ordered pair $c = (c_1, c_2)$, where $c_1, c_2 \in \{1, \ldots, 2 \sqrt{\Delta}\}$. The first coordinate $c_1$ is analogous to the notion of a {\em type}, as defined in the previous algorithm.  For any vertex $v \in V$ and $j \in \{1, 2\}$, let $\chi^*_j(v)$ denote the $j$-tuple consisting of the first $j$ coordinates of the color assigned to $v$. For ease of exposition, we define $\chi^*_0(v) = \bot$. Furthermore, for every vertex $v \in V$ and every $j \in \{0,1,2\}$, let $N^*_j(v) = \{ u \in N_v : \chi^*_j(u) = \chi^*_j(v) \}$ denote the set of neighbors $u$ of $v$ with $\chi^*_j(u) = \chi^*_j(v)$. With these notations, Property~\ref{prop:deterministic} can be rewritten as:  $|N^*_1(v)| < 2 \sqrt{\Delta}$ for all $v \in V$. 

To improve the amortized update time to $O(\polylog \Delta)$, we think of every color as an $L$ tuple $c = (c_1, \ldots, c_L)$, whose { each coordinate can take $\lambda$  possible values. The total number of colors is given by $\lambda^L = |\C|$. The values of $L$ and $\lambda$ are chosen in such a way which ensures that $\lambda = O(\lg^{1+o(1)} \Delta)$ and $L = O(\lg \Delta/\lg \lg \Delta)$.  We maintain the invariant that $|N^*_j(v)| \leq (\Delta/\lambda^j) \cdot f(j)$ for all $v \in V$ and $j \in [0, L]$, for some carefully chosen function $f(j)$. We then implement a generalization of the previous algorithm on these colors represented as $L$ tuples. Using some carefully chosen parameters we show how to deterministically maintain a $(\Delta+o(\Delta))$ vertex  coloring in a dynamic graph in $O(\polylog \Delta)$ amortized update time. See Section~\ref{sec:det:coloring} for the details.

\section{A Randomized Dynamic Algorithm for $\Delta+1$ Vertex Coloring}
\label{sec:randomized:coloring}

As discussed in Section~\ref{sec:idea:randomized},  our randomized dynamic algorithm for $\Delta+1$ vertex  coloring has two main components. 
The first one is  a hierarchical partition of the vertices   of the input graph into $O(\log \Delta)$-many levels. In Section~\ref{sec:hierarchy}, we show how to maintain such a hierarchical partition dynamically. The second component is the use of randomization while recoloring a conflicted vertex $v$  so as to ensure that (a) at most one new conflict is caused due to this recoloring, and (b) if so, the new conflicted vertex lies at a  level strictly lower than $\ell(v)$. We describe this second component  in Section~\ref{sec:recoloring}. 
The complete algorithm, which combines the two components,  appears in Section~\ref{sec:randomized:complete}.  The theorem below captures our main result.

\begin{theorem}
	\label{th:randomized:main}
	There is a randomized, fully dynamic algorithm to maintain a $\Delta+1$ vertex  coloring of a graph whose maximum degree is $\Delta$
	with expected amortized update time $O(\log \Delta)$.
\end{theorem}


\subsection{Preliminaries.}
\label{sec:randomized:prelim}
\label{sec:randomized:definitions}
\label{sec:randomized:data:structures}

We  start with the definition of a hierarchical partition.  Let $G = (V, E)$ denote the input graph that is changing dynamically, and let $\Delta$ be an upper bound on the maximum degree of any vertex in $G$. For now we assume that the value of $\Delta$ does not change with time. In Section~\ref{sec:general}, we explain how to relax this assumption. Fix a constant $\beta > 20$. For simplicity of exposition, assume  that $\log_{\beta} \Delta = L$ (say) is an integer and that  $L > 3$. 
The vertex set $V$ is partitioned into $L -3$ subsets  $V_4, \ldots, V_L$. The {\em level} $\ell(v)$ of a vertex $v$ is the index of the subset it belongs to.
For any vertex $v \in V$ and any two indices $4 \leq i \leq j \leq L$, we let $\N_v(i, j) = \{ u : (u, v) \in E, ~~ i\leq \ell(u)\leq j \}$ 
be the set of neighbors of $v$ whose levels are between $i$ and $j$. For notational convenience, we define $\N_v(i, j) = \emptyset$ whenever $i > j$.
A hierarchical partition satisfies the following two properties/invariants. Note that since $\beta^L = \Delta$,  Invariant~\ref{inv:hierarchy:2} is trivially satisfied by every vertex at the highest level $L$. Invariant~\ref{inv:hierarchy:1}, on the other hand, is trivially satisfied by the  vertices at  level $4$.

\begin{invariant}
	\label{inv:hierarchy:1}
	For every vertex $v \in V$ at level $\ell(v) > 4$, we have $|\N_v(4, \ell(v) -1)| \geq \beta^{\ell(v) - 5}$. 
\end{invariant}

\begin{invariant}
	\label{inv:hierarchy:2}
	For every vertex $v \in V$, we have $|\N_v(4, \ell(v))| \leq \beta^{\ell(v)}$.
\end{invariant}

Let $\C = \{1, \ldots, \Delta+1\}$ be the set of all possible colors. 
A coloring $\chi:V\to \C$ is {\em proper} for the graph $G = (V, E)$ iff for every edge $(u,v) \in E$, we have $\chi(u)\neq \chi(v)$.
Given the hierarchical partition, a coloring $\chi:V\to \C$, and a vertex  $x$ at level $i=\ell(x)$, we define a few key subsets of $\C$.
Let $\C^+_x := \bigcup_{y\in \N_x(i,L)} \chi(y)$ be the colors used by neighbors of $x$ lying in levels $i$ and above. 
Let $\C_x := \C \setminus \C^+_x$ denote the remaining set of colors. 
We say a color $c\in \C_x$ is {\em blank for $x$} if no vertex in $\N_x(4,i-1)$ is assigned color $c$.
We say a color $c\in \C_x$ is {\em unique for $x$} if {\em exactly one} vertex in $\N_x(4,i-1)$ is assigned color $c$.
We let $B_x$ (respectively $U_x$) denote the blank (respectively unique) colors for $x$. 
Let $T_x := \C_x \setminus (B_x \cup U_x)$ denote the remaining colors in $C_x$. Thus, for every color $c \in T_x$, there are at least two vertices   $u \in \N_x(4,i-1)$ that are assigned color $c$.
We end this section with a crucial observation.

\begin{claim}\label{cl:blank-unique-size}
	For any vertex  $x$ at level $i$, we have $|B_x \cup U_x| \geq 1 + \frac{|\N_x(4,i-1)|}{2}$.
\end{claim}
\begin{proof}
	Since $|\C| = 1+ \Delta \geq 1 + |\N_x(4,L)|$ and $|\C^+_x| \leq |\N_x(i,L)|$, we get $|\C_x| \geq 1 + |\N_x(4,i-1)|$.
	The following two observations, which in turn follow from definitions, prove the claim; (a) $|\C_x| = |B_x\cup U_x| + |T_x|$ and (b) $2 |T_x| \leq |\N_x(4,i-1)|$.
\end{proof}

\noindent {\bf Data Structures.} We now describe the data structures used by our dynamic algorithm. The first set is used to maintain the hierarchical partition and the second set is used to maintain the sets of colors.

\smallskip
\noindent (1) For every vertex $v\in V$ and every level $\ell(v) \leq i \leq L$, we maintain the neighbors $\N_v(i,i)$ of $v$ in level $i$ in a doubly linked list. If $\ell(v) > 4$, then we also maintain the set of neighbors $\N_v(4,\ell(v) - 1)$ in a doubly linked list.  We use the phrase {\em neighborhood list} of $v$ to refer to any one of these lists. For every neighborhood list we  maintain a counter which stores the number of vertices   in it. Every edge $(u,v) \in E$ keeps two pointers -- one to the position of $u$ in the neighborhood list of $v$, and the other vice versa. Therefore when an edge is inserted into or deleted from $G$ the linked lists can be updated in $O(1)$     time. Finally, we keep two queues of {\em dirty} vertices   which store the vertices   not satisfying either of the two invariants.
	
\smallskip
\noindent (2) We maintain the coloring $\chi$ as an array where $\chi(v)$ contains the current color of $v$.
	Every vertex $v$ maintains the colors $\C^+_v$ and $\C_v$ in doubly linked lists. 
	For each color $c$ and vertex $v$, we keep a pointer from the color to its position in either $\C^+_v$ or $\C_v$ depending on which list $c$ belongs to. This allows us to add and delete colors from these lists in $O(1)$ time.
	We also maintain  a counter $\mu_v^+(c)$  associated with each color $c$ and each vertex $v$. If $c \in \C^+_v$, then the value of $\mu_v^+(c)$ equals the number of neighbors $y \in \N_v(i, L)$ with color $\chi(y) = c$. Otherwise, if $c \in \C_v$, then we set $\mu_v^+(c) \leftarrow 0$.
	For each vertex $v$, we keep a {\em time} counter $\tau_v$ which stores the last ``time'' (edge insertion/deletion) at which $v$ was recolored\footnote{Note that as long as the number of edge insertions and deletions are polynomial, $\tau_v$ requires only $O(\log n)$ bits to store; if the number becomes superpolynomial then every $n^3$ rounds or so we recompute the full coloring in the current graph.}, i.e., its $\chi(v)$ was changed.

\subsection{Maintaining the hierarchical partition.}
\label{sec:hierarchy}

Initially when the graph is empty,  all the vertices  are at level $4$. This satisfies both the invariants vacuously.
Subsequently, we ensure that the hierarchical partition satisfies Invariants~\ref{inv:hierarchy:1} and~\ref{inv:hierarchy:2} by using a simple greedy heuristic procedure. To describe this procedure, we define a vertex to be {\em dirty} if it violates any one of the invariants, and {\em clean} otherwise. Our goal then is to ensure that every vertex in the hierarchical partition remains clean. By inductive hypothesis, we assume that every vertex is clean before the insertion/deletion of an edge. Due to the insertion/deletion of an edge $(u, v)$, some vertices   of the form $x \in \{u, v\}$ might become dirty.  We fix the dirty vertices   as per the procedure described in Figure~\ref{fig:hierarchy}. In this procedure,  we always fix the dirty vertices   violating Invariant~\ref{inv:hierarchy:2} before fixing any dirty vertex that violates Invariant~\ref{inv:hierarchy:1}. This will be crucial in bounding the amortized update time. Furthermore, note that as we change the level of a vertex $x$ during one iteration of the {\sc While} loop in Figure~\ref{fig:hierarchy}, this might lead to a change in the {\em below or side degrees}\footnote{The terms {\em below-degree} and {\em side-degree} of a vertex $v$  refer to the values of  $|\N_v(4, \ell(v)-1)|$ and $|\N_v(\ell(v), \ell(v))|$ respectively.} of the neighbors of $x$. Hence, one iteration of the {\sc While} loop might create multiple new dirty vertices, which are dealt with in the subsequent iterations of the same {\sc While} loop. 
It is not  hard to see that any iteration of the while loop acting on a vertex  $x$ ends with making it clean. 
We encapsulate this in the following lemma and the subsequent corollary.

\begin{figure}[h!]
	\centerline{\framebox{
			\begin{minipage}{5.5in}
				\begin{tabbing}
					01.     {\sc While}  Invariant~\ref{inv:hierarchy:1} or Invariant~\ref{inv:hierarchy:2} is violated: \\
					 02.   \ \ \ \ \   \=  {\sc If} there is some vertex $x \in V$   that violates Invariant~\ref{inv:hierarchy:2}, {\sc Then} \\
					 03.  \> \ \ \ \ \  \= Find the minimum level $k > \ell(x)$  where $|\N_x(4, k)| \leq \beta^{k}$. \\
					 04.  \> \> Move the vertex $x$ up to level $k$, and  update the relevant data structures as \\
					 \>  \> described in the proof of Lemma~\ref{lm:hierarchy:time:up}. \\
					 05.  \> {\sc Else} \\
					 06.  \> \> Find a vertex $x \in V$ that  violates Invariant~\ref{inv:hierarchy:1}. \\
					 07.  \> \> {\sc If} there is a level $4 < k < \ell(x)$   where $|\N_x(4, k-1)| \geq \beta^{k-1}$, {\sc Then} \\
				         08. \> \>  \ \ \ \ \ \ \ \= Move  the vertex $x$ down   to {\em maximum} such level $k$,   and update the relevant  \\
					 \> \> \> data structures as described   in the proof of  Lemma~\ref{lm:hierarchy:time:down}. \\
					 09.  \> \> {\sc Else}  \\
					 10.  \> \> \> Move the vertex $x$ down  to level $4$, and update the relevant data structures.
				\end{tabbing}
			\end{minipage}
	}}
	\caption{\label{fig:hierarchy} Subroutine: MAINTAIN-HP is called when an edge $(u,v)$ is inserted into or deleted from $G$.}
\end{figure}

\begin{lemma}
	\label{lm:while:loop:clean}
	Consider any iteration of the {\sc While} loop in Figure~\ref{fig:hierarchy} which changes the level of a vertex $x \in V$ from $i$ to $k$. The vertex $x$ becomes clean (i.e., satisfies both the invariants) at the end of the iteration. Furthermore, at the end of this iteration we have: (a) $|\N_x(4, k)| \leq \beta^k$,  (b) $|\N_x(4, k-1)| \geq \beta^{k-1}$ if $k > 4$. 
\end{lemma}
\begin{proof}
 There are three  cases to consider, depending on how the vertex moves from  level $i$ to  level $k$.
	
	\smallskip
	\noindent {\em Case 1. The vertex $x$ moves up from a level $i \in [4, L-1]$.}  In this case, the vertex moves up to the {\em minimum}  level $k > i$ where $|\N_x(4, k)| \leq \beta^k$. This implies that  $|\N_x(4, k-1)| > \beta^{k-1} > \beta^{k-5}$. Thus, the vertex $x$ satisfies both the invariants after it moves to level $k$, and both the conditions (a) and (b) hold.
	
	\smallskip
	\noindent {\em Case 2. The vertex $x$ moves down from level $i$ to a level $4 < k < i$.} In this case, steps 07, 08 in Figure~\ref{fig:hierarchy} imply that $k < i$ is the {\em maximum} level where $|\N_x(4, k-1)| \geq \beta^{k-1}$. Hence, we have $|\N_x(4, k)| < \beta^k$. So the vertex satisfies both the  invariants after moving to level $k$, and  both the conditions (a) and (b) hold.
	
	\smallskip
	\noindent {\em Case 3. The vertex $x$ moves down from level $i$ to level $k=4$.} Here, steps 07, 09 in Figure~\ref{fig:hierarchy} imply that $|\N_x(4, j-1)| < \beta^{j-1}$  for every level $4 < j < i$.  In particular, setting $j = 5 = k+1$, we get: $|\N_x(4, k)|  < \beta^0 = 1 < \beta^4$. Thus, the vertex  satisfies both the invariants after it moves down to level $4$, and both the conditions (a) and (b) hold.
\end{proof}


Lemma~\ref{lm:while:loop:clean}  states that during any given iteration of the {\sc While} loop in Figure~\ref{fig:hierarchy}, we pick a dirty vertex $x$ and make it clean. In the process, some neighbors of $x$ become dirty, and they are handled in the subsequent iterations of the same {\sc While} loop. When the {\sc While} loop terminates, every vertex is clean by definition. It now remains to analyze the  time spent on implementing this {\sc While} loop after an edge insertion/deletion in the input graph. Lemma~\ref{lm:while:loop:clean} will be crucial in this analysis. The intuition is as follows. The lemma guarantees that whenever a vertex $x$ moves to a level $k > 4$, its below-degree is at least $\beta^{k-1}$. In contrast, Invariant~\ref{inv:hierarchy:1} and Figure~\ref{fig:hierarchy} ensure that whenever the  vertex moves down from the same level $k$, its below-degree is less than $\beta^{k-5}$. Thus, the vertex loses at least $\beta^{k-1} - \beta^{k-5}$ in below-degree before it  moves down from level $k$. This {\em slack} of $\beta^{k-1} - \beta^{k-5}$  help us bound the amortized update time. 
We next bound the time spent on a single iteration of the {\sc While} loop in Figure~\ref{fig:hierarchy}.

\begin{lemma}
	\label{lm:hierarchy:time:up}
	Consider any iteration of the {\sc While} loop in Figure~\ref{fig:hierarchy} where a vertex $x$ moves up to a level $k$ from a level $i < k$ (steps 2 -- 4). It takes $\Theta(\beta^k)$ time to implement such an iteration.
\end{lemma}

\begin{proof}
	First, we claim that  the value of $k$ (the level where the vertex $x$ will move up to) can be identified in $\Theta(k-i)$ time. 
	This is because we explicitly store the sizes of the lists $\N_x(4, i-1)$ and $\N_x(j,j)$ for all $j \geq i$. Next, we update the lists $\{\C^+_v, \C_v\}$ and the counters $\{\mu^+_v(c)\}$ for  $x$ and its neighbors as follows.
	\medskip
	\noindent {\sc For} every level  $j \in \{i, \ldots, k\}$  and every vertex $y\in \N_x(j,j)$:
	\begin{itemize}
		\item $\C^+_y \leftarrow \C^+_y \cup \{ \chi(x) \}$, $\C_y \leftarrow \C_y \setminus \{ \chi(x) \}$ and $\mu^+_y(\chi(x)) \leftarrow \mu^+_y(\chi(x)) + 1$. 
		\item {\sc If} $j < k$, {\sc Then}
		\begin{itemize}
			\item  $\mu^+_x(\chi(y)) \leftarrow \mu^+_x(\chi(y)) - 1$. 
			\item If $\mu^+_x(\chi(y)) = 0$, then  $\C^+_x \leftarrow \C^+_x \setminus \{ \chi(y) \}$ and $\C_x \leftarrow \C_x \cup \{ \chi(y) \}$. 
		\end{itemize}
	\end{itemize}
	
	The  time spent on the above operations is bounded by the number of vertices   in $\N_x(4,k)$.

	Since the vertex $x$ is moving up from level $i$ to level $k > i$, we have to update the position of $x$ in the neighborhood lists of the vertices   $u \in \N_x(4, k)$. We also need to merge the lists $\N_x(4, i-1)$ and $\N_x(j, j)$ for $i \leq j < k$ into a single list $\N_x(4, k-1)$. In the process if some vertices    $u \in \N_x(4, k)$ becomes dirty, then we need to put them in the correct dirty queue. This takes   $\Theta(|\N_x(4, k)|)$ time. 
	
By Lemma~\ref{lm:while:loop:clean},  we have $|\N_x(4, k)| \leq \beta^k$, and $|\N_x(4, k)| \geq |\N_x(4, k-1)| \geq \beta^{k-1}$. Since $\beta$ is a constant, we conclude that it takes $\Theta(\beta^k)$ time to implement this iteration of the {\sc While} loop in Figure~\ref{fig:hierarchy}.
\end{proof}

\begin{lemma}
	\label{lm:hierarchy:time:down}
	Consider any iteration of the {\sc While} loop in Figure~\ref{fig:hierarchy} where a vertex $x$ moves down to a level $k$ from a level $i > k$ (steps 5 -- 10). It takes $O(\beta^i)$ time to implement such an iteration.
\end{lemma}

\begin{proof}
	We first  bound the time spent on identifying the level $k < i$ the vertex $x$ will move down to. 
	Since the vertex $x$ violates Invariant~\ref{inv:hierarchy:1}, we know that $|\N_x(4,i-1)| < \beta^{i-5} = O(\beta^i)$. 
	Therefore, the algorithm can scan through the list $\N_x(4,i-1)$  and find the required level $k$
	in $\Theta(i + |\N_x(4,i-1)|)$ time. Next, we update the lists $\{\C^+_v, \C_v\}$ and the counters $\{\mu^+_v(c)\}$ for  $x$ and its neighbors as follows.
	
	\medskip
	\noindent {\sc For}  every vertex $y\in \N_x(4,i-1) \cup \N_x(i, i)$:
	\begin{itemize}
		\item {\sc If} $i \geq \ell(y) > k$, {\sc Then}
		\begin{itemize}
			\item  $\mu^+_y(\chi(x)) \leftarrow \mu^+_y(\chi(x)) - 1$. 
			\item If $\mu^+_y(\chi(x)) = 0$, then $\C^+_y \leftarrow \C^+_y \setminus \{ \chi(x) \}$ and $\C_y \leftarrow \C_y \cup \{ \chi(x) \}$. 
		\end{itemize}
		\item {\sc If} $i > \ell(y) \geq k$, {\sc Then}
		\begin{itemize}
			\item $\C^+_x \leftarrow \C^+_x \cup \{\chi(y)\}$, $\C_x \leftarrow \C_x \setminus \{ \chi(y)\}$ and $\mu^+_x(\chi(y)) \leftarrow \mu^+_x(\chi(y)) + 1$. 
		\end{itemize}
	\end{itemize}
	The time spent on the above operations is bounded by the number of vertices   in $\N_x(4, i)$.	
	
	Since the vertex $x$ is moving down from level $i$ to level $k < i$, we have to update the position of $x$ in the neighborhood lists of the vertices   $u \in \N_x(4, i)$. We also need to split the list $\N_x(4, i-1)$ into the lists $\N_x(4, k-1)$ and $\N_x(j, j)$ for $k \leq j < i$.  In the process if some vertices   $u \in \N_x(4, i)$ become dirty, then we need to put them in the correct dirty queue. This  takes $\Theta(|\N_x(4, i)|)$ time. 
	
 Figure~\ref{fig:hierarchy} ensures that $x$ satisfies Invariant~\ref{inv:hierarchy:2} at level $i$ before it moves down to a lower level. Thus, we have $|\N_x(4,i)| \leq \beta^i$, and we spend $\Theta(1+ |\N_x(4, i)|) = O(\beta^i)$ time  on this iteration of the {\sc While} loop.  
\end{proof}

\begin{corollary}
	\label{cor:lm:hierarchy:time}
	It takes $\Omega(\beta^{k})$ time for a vertex $x$ to move from a level $i$ to a different level $k$.
\end{corollary}

\begin{proof}
	If $4 \leq i < k$, then the corollary follows immediately from Lemma~\ref{lm:hierarchy:time:up}. For the rest of the proof, suppose that $i > k$. In this case, as per the proof of Lemma~\ref{lm:hierarchy:time:down}, the time spent is at least the size of the list $\N_x(4, k-1)$, and Lemma~\ref{lm:while:loop:clean} implies that $|\N_x(4, k-1)| \geq \beta^{k-1}$. Hence, the total time spent is $\Omega(\beta^{k-1})$, which is also $\Omega(\beta^k)$ since $\beta$ is a constant.  Note that we ignored the scenarios where $k = 4$ since in that event $\beta^k$ is a constant anyway.
\end{proof}

In Theorem~\ref{th:hierarchy:maintain}, we bound the amortized update time for maintaining a hierarchical partition.

\begin{theorem}
	\label{th:hierarchy:maintain}
	We can maintain a hierarchical partition of the vertex set $V$ that satisfies Invariants~\ref{inv:hierarchy:1} and~\ref{inv:hierarchy:2} in $O(\log \Delta)$ amortized update time.
\end{theorem}
We devote the rest of Section~\ref{sec:hierarchy} to the proof of the above theorem using a token based scheme. 
The basic framework is as follows. For every edge insertion/deletion in the input graph we create at most $O(L)$ tokens, and we use one token to perform $O(\beta^2)$ units of computation. This implies an amortized update time of $O(\beta^2 \cdot L) = O(\beta^2 \cdot \log_{\beta} \Delta)$, which is $O(\log \Delta)$ since $\beta$ is a constant. 

Specifically, we associate $\theta(v)$ many tokens with every vertex $v \in V$  and $\theta(u, v)$ many tokens with every edge $(u, v) \in E$  in the input graph. The values of these tokens are determined by the following equalities.
\begin{eqnarray}
\label{eq:token:edge}
\theta(u, v) & = & L - \max(\ell(u), \ell(v)). 
\end{eqnarray}
\begin{eqnarray}
\label{eq:token:node}
 \ \ \ \  \ \theta(v) &  = & \frac{\max\left(0, \beta^{\ell(v)-1} - |\N_v(4, \ell(v)-1)|\right)}{2\beta}   \qquad \ \ \ \ {\rm if } \ \ell(v) > 4; \\
 & = & 0 \qquad \qquad \qquad \qquad \qquad \qquad \qquad \qquad \qquad  {\rm  otherwise.} \nonumber
\end{eqnarray} 
Initially, the input graph $G$ is empty, every vertex is at level $4$, and  $\theta(v) = 0$ for all $v \in V$. Due to the insertion of an edge $(u, v)$ in $G$, the total number of tokens  increases by at most $L - \max(\ell(u), \ell(v)) < L$, where $\ell(u)$ and $\ell(v)$ are the levels of the endpoints of the edge just before the insertion. On the other hand, due to the deletion of an edge $(u, v)$ in the input graph, the value of $\theta(x)$ for $x \in \{u, v\}$ increases by at most $1/(2\beta)$, and the tokens associated with the edge $(u, v)$ disappears. Overall,  the total number of tokens increases by at most $1/(2\beta) + 1/(2\beta) \leq O(L)$ due to the deletion of an edge. We now show that the work done during one iteration of the {\sc While} loop in Figure~\ref{fig:hierarchy} is proportional to $O(\beta^2)$ times the net decrease in the total number of tokens during the same iteration. Accordingly, we focus on any single iteration of the {\sc While} loop in Figure~\ref{fig:hierarchy} where a vertex $x$ (say) moves from level $i$ to level $k$. We consider two cases, depending on whether  $x$ moves to a higher or a lower level. 

\paragraph{Case 1: The vertex $x$ moves up from level $i$ to level $k > i$.} 
\label{sub:sec:case1}
Immediately after the vertex $x$ moves up to level $k$, we have $\N_x(4, k-1) \geq \beta^{k-1}$ and hence $\theta(x) = 0$. This follows from~(\ref{eq:token:node}) and Lemma~\ref{lm:while:loop:clean}. Since $\theta(x)$ is always nonnegative,  the value of $\theta(x)$ does not increase as $x$ moves up to level $k$. We now focus on bounding the change in the total number of tokens associated with the neighbors of $x$.  Note that the event of $x$ moving up from level $i$ to level $k$ affects only the tokens associated with the vertices   $u \in \N_x(4, k)$. Specifically, from~(\ref{eq:token:node}) we infer that for every vertex $u \in \N_x(4, k)$, the value of $\theta(u)$ increases by at most $1/(2\beta)$. On the other hand, for every vertex $u\in \N_x(k+1, L)$, the value of $\theta(u)$ remains unchanged. Thus, the total number of tokens associated with the neighbors of $x$ increases by at most $(2\beta)^{-1} \cdot |\N_x(4, k)| \leq (2\beta)^{-1} \cdot \beta^k = \beta^{k-1}/2$. The inequality follows from Lemma~\ref{lm:while:loop:clean}. To summarize, the total number of tokens associated with all the vertices   increases by at most $\beta^{k-1}/2$. 

We now focus on bounding the change in the total number of tokens associated with the edges incident on $x$. From~(\ref{eq:token:edge}) we infer that for every edge $(x, u)$ with $u\in \N_x(4, k-1)$, the value of $\theta(x, u)$ drops by at least one as the vertex $x$ moves up from level $i < k$ to level $k$. For every other edge $(x, u)$ with $u \in \N_x(k, L)$, the value of $\theta(x, u)$ remains unchanged. Overall, this means that the total number of tokens associated with the edges drops by at least $|\N_x(4, k-1)| \geq \beta^{k-1}$. The inequality follows from Lemma~\ref{lm:while:loop:clean}. To summarize, the total number of tokens associated with the edges decreases by at least $\beta^{k-1}$. 

From the discussion in the preceding two paragraphs, we reach the following conclusion: As the vertex $x$ moves up from level $i < k$ to level $k$, the total number of tokens associated with all the vertices   and edges decreases by at least $\beta^{k-1} - \beta^{k-1}/2 \geq \beta^{k-1}/2$. In contrast, Lemma~\ref{lm:hierarchy:time:up} states that it takes $O(\beta^k)$ time taken to implement this iteration of the {\sc While} loop in Figure~\ref{fig:hierarchy}. Hence,  we derive that the time spent on updating the relevant data structures is at most $O(2 \beta) = O(\beta^2)$ times the net decrease in the total number of tokens. This concludes the proof of Theorem~\ref{th:hierarchy:maintain} for Case 1.

\paragraph{Case 2: The vertex $x$ moves down from level $i$ to level $k < i$.}
\label{sub:sec:case2}
As in Case 1, we begin by observing that immediately after the vertex $x$ moves down to level $k$, we have $|\N_x(4, k-1)| \geq \beta^{k-1}$ if $k > 4$, and hence $\theta(x) = 0$. This follows from Lemma~\ref{lm:while:loop:clean}. The vertex $x$ violates Invariant~\ref{inv:hierarchy:1} just before moving from level $i$ to level $k$ (see step 6 in Figure~\ref{fig:hierarchy}). In particular, just before the vertex moves down from level $i$ to level $k$, we have $|\N_x(4, i-1)| < \beta^{i-5}$ and $\theta(x) \geq (2\beta)^{-1} \cdot \left(\beta^{i-1} - \beta^{i-5}\right) = \beta^{i-2}/2 - \beta^{i-6}/2 \geq \beta^{i-2}/3$. The last inequality holds since $\beta$ is a sufficiently large constant. So the number of tokens associated with  $x$ drops by at least $\beta^{i-2}/3$ as it moves down from level $i$ to level $k$. Also, from~(\ref{eq:token:node}) we infer that the value of $\theta(u)$ does not increase for any  $u \in \N_x$ as $x$ moves down to a lower level. Hence, we conclude that:
\begin{eqnarray}
\label{eq:text:down}
{\rm   The \ total \ number \ of \ tokens \ associated } 
{\rm \ with \ all \ the \ vertices  \ drops \ by \ at \ least \ } \beta^{i-2}/3. 
\end{eqnarray}

We now focus on bounding the change in the number of tokens associated with the edges incident on $x$. From~(\ref{eq:token:edge}) we infer that the number of tokens associated with an edge $(u, x) \in \N_x(4, i-1)$ increases by $(i - \max(k, \ell(u)))$ as  $x$ moves down from level $i$ to level $k$. In contrast, the number of tokens associated with any other edge $(u, x) \in \N_x(i, L)$ does not change as the vertex $x$ moves down from level $i$ to a lower level. Let $\Gamma$ be the increase in the total number of tokens associated with all the edges. Thus, we have: 
\begin{eqnarray}
\label{eq:gamma:1}
\Gamma   =  \sum_{(u, x) \in \N_x(4, i-1)} (i- \max(k, \ell(u)))  
 =  \sum_{j= k}^{i-1} |\N_x(4,j)|,
\end{eqnarray}
where the last equality follows by rearrangement. Next, recall that the vertex $x$ moves down from level $i$ to level $k$ during the concerned iteration of the {\sc While} loop in Figure~\ref{fig:hierarchy}. Accordingly, steps 7 -- 10 in Figure~\ref{fig:hierarchy} implies that $|\N_x(4, j-1)| < \beta^{j-1}$  for all levels $i > j > k$. This is equivalent to the following statement: 
\begin{equation} 
\label{eq:text:down:1}
|\N_x(4, j)| < \beta^j  {\rm  \ for \ all \ levels \ } i-1 > j \geq k.
\end{equation} 
Next, step 6 in Figure~\ref{fig:hierarchy} implies that the vertex $x$ violates Invariant~\ref{inv:hierarchy:1} at level $i$. Thus, we get: $|\N_x(4, i-1)| < \beta^{i-5}$. Note that for all levels $j < i$, we have $\N_x(4, j) \subseteq \N_x(4, i-1)$ and  $|\N_x(4, j)| \leq |\N_x(4, i-1)|$. Hence, we get: $|\N_x(4, j)| < \beta^{i-5}$ for all levels $j < i$.  Combining this observation with~(\ref{eq:text:down:1}), we get:
\begin{eqnarray}
\label{eq:gamma:3}
|\N_x(4, j)| < \min(\beta^{i-5}, \beta^j)  {\rm \  for  \ all  \ levels \ } i-1 \geq j \geq k.
\end{eqnarray}
Plugging~(\ref{eq:gamma:3}) into~(\ref{eq:gamma:1}), we get:
\begin{equation}
\label{eq:gamma:4}
\Gamma < 5 \beta^{i-5} + \sum_{j=k}^{i-6} \beta^j < \beta^{i-3}.
\end{equation}
In the above derivation, the last inequality holds since $\beta$ is a sufficiently large constant.

From~(\ref{eq:text:down}) and~(\ref{eq:gamma:4}), we reach the following conclusion: As the vertex $x$ moves down from level $i$ to a level $k < i$, the total number of tokens associated with all the vertices   and edges decreases by at least $\beta^{i-2}/3 - \Gamma > \beta^{i-2}/3 - \beta^{i-3} = \Omega(\beta^{i-2})$. In contrast, by Lemma~\ref{lm:hierarchy:time:down} it takes $O(\beta^i)$ time to implement this iteration of the {\sc While} loop in Figure~\ref{fig:hierarchy}. Hence,    we derive that the time spent on updating the relevant data structures is at most $O(\beta^2)$ times the net decrease in the total number of tokens. This concludes the proof of Theorem~\ref{th:hierarchy:maintain} for Case 2.

\subsection{The recoloring subroutine.}
\label{sec:recoloring}
Whenever we want to change the color of a vertex $v \in V$, we call the subroutine RECOLOR($v$) as described in Figure~\ref{fig:recolor}. We ensure that the hierarchical partition does not change during a call to this subroutine. Specifically, throughout the duration of any call to the RECOLOR subroutine, the value of $\ell(x)$ remains the same for every vertex $x \in V$. We also ensure that the hierarchical partition satisfies Invariants~\ref{inv:hierarchy:1} and~\ref{inv:hierarchy:2} before any making any call to the RECOLOR subroutine.

During a call to the subroutine RECOLOR$(v)$, we  randomly choose a color $c$ for  the vertex $v$ from the subset $B_v \cup U_v \subseteq \C_v$. In case the random color $c$  lies in $U_v$, we find the unique neighbor $v' \in \N_v(4,\ell(v)-1)$ of $v$ which is assigned this color, and then we recursively recolor $v'$. Since the level of $v'$ is strictly less than that of $v$, the maximum depth of this recursion is $L$.  We now  bound the time spent on a  call to RECOLOR$(v)$.

\begin{figure}[h!]
	\centerline{\framebox{
			\begin{minipage}{5.5in}
				\begin{tabbing}
					 1. \=  Choose $c\in B_v \cup U_v$ uniformly at random. \ \ \  // {\em These notations are defined in Section~\ref{sec:randomized:definitions}.} \\
					 2.  \>  Set $\chi(v) \leftarrow c$. \\
					 3. \>    Update the relevant data structures as  described in the proof of Lemma~\ref{lm:fix:time}. \\					
					 4.   \> {\sc If} $c\in U_v$: \\
					 5.   \> \ \ \ \ \ \ \ \ \= Find the {\em unique} vertex  $v' \in \N_v(4,\ell(v)-1)$  with $\chi(v') = c$.\\
					 6.   \> \>  RECOLOR($v'$).
				\end{tabbing}
			\end{minipage}
	}}
	\caption{\label{fig:recolor} Subroutine RECOLOR($v$)  }
\end{figure}

\begin{lemma}
	\label{lm:fix:time}
	It takes  $O(\beta^{\ell(v)})$ time to implement one call to RECOLOR$(v)$. This includes the total time spent on the chain of subsequent recursive calls that originate from the call to RECOLOR$(v)$.
\end{lemma}	

\begin{proof}
	Let us assume that $\ell(v) = i$ and $\chi(v) = c'$ just before the call to  RECOLOR$(v)$.
	To implement Step 01 in Figure~\ref{fig:recolor}, the vertex $v$ scans the neighborhood list $\N_v(4,i-1)$ and computes the subset colors $T_v \subseteq \C_v$ which appear twice or more among the these vertices. The vertex $v$ keeps these colors $T_v$ in a separate list and deletes every color in $T_v$ from the list $\C_v$. On completion, the list $\C_v$ consists of the colors in $B_v\cup U_v$ and the algorithm samples a random color $c$ from this list.\footnote{Note that we might have $|B_v \cup U_v| \gg \beta^{\ell(v)}$ and so it is not  clear how to sample in $O(\beta^{\ell(v)})$ time. The modification required here is that it is sufficient to sample from the first $\beta^{\ell(v)}$ elements of $B_v\cup U_v$. For clarity of exposition, we ignore this issue.} Next, the algorithm adds all the colors in $T_v$ back to the list $\C_v$, thereby restoring the list $\C_v$ to its actual state. The algorithm can also do another scan of $\N_v(4,i-1)$ to check whether $c\in B_v$ or $c\in U_v$. 
	The total time taken to do all this is $\Theta(|\N_v(4,i-1)|)$ which by Invariants~\ref{inv:hierarchy:1} and~\ref{inv:hierarchy:2} is $\Theta(\beta^{\ell(v)})$.
	After changing  the color of the vertex $v$ from $c'$ to $c$ in step 02, the algorithm needs to update the  data structures (see Section~\ref{sec:randomized:data:structures}) as follows.
	\begin{itemize}
		\item {\sc For} every vertex $w\in \N_v(4,i)$:
		\begin{itemize}
			\item $\mu^+_w(c') \leftarrow \mu^+_w(c') - 1$.
			\item {\sc If} $\mu^+_w(c') = 0$, {\sc Then} $\C^+_w \leftarrow \C^+_w \setminus \{c'\}$ and $\C_w \leftarrow \C_w \cup \{c'\}$. 
			\item  $\C^+_w \leftarrow \C^+_w \cup \{c\}$, $\C_w \leftarrow \C_w \setminus \{c\}$ and $\mu^+_w(c) \leftarrow \mu^+_w(c)+1$. 
		\end{itemize}
	\end{itemize}
	The above operations also take $\Theta(|\N_v(4,i)|) = \Theta(\beta^{\ell(v)})$ time, as per Invariants~\ref{inv:hierarchy:1} and~\ref{inv:hierarchy:2}. 
	
	Finally, in the subsequent recursive calls suppose we recolor the vertices   $y_1,y_2,\ldots$. Note that $\ell(x) > \ell(y_1) > \ell(y_2) > \cdots $. Therefore the total time taken can be bounded by $\Theta(\beta^{\ell(x)} + \beta^{\ell(y_1)} + \cdots) = \Theta(\beta^{\ell(x)})$ since it is a geometric series sum.
	This completes the proof. 
\end{proof}

\subsection{The complete algorithm and analysis.}
\label{sec:randomized:complete}

Initially, when the  graph $G = (V, E)$ is empty, every vertex $v \in V$ belongs to level $4$ and picks a random color $\chi(v) \in \C$. At this point, the coloring $\chi$ is proper since there are no edges, and Invariants~\ref{inv:hierarchy:1} and~\ref{inv:hierarchy:2} are vacuously satisfied. Now, by inductive hypothesis, suppose that before the insertion or deletion of an edge in $G$, we have the guarantee that: (1) $\chi$ is a proper coloring and (2) Invariants~\ref{inv:hierarchy:1},~\ref{inv:hierarchy:2} are satisfied. We handle the insertion or deletion of this edge in $G$ according to the procedure  in Figure~\ref{fig:recolor:main}.

Specifically, after the insertion or deletion of an edge $(u,v)$, we first update the hierarchical partition by calling the subroutine MAINTAIN-HP (see Figure~\ref{fig:hierarchy}). At the end of the call to this subroutine, we know for sure that Invariants~\ref{inv:hierarchy:1} and~\ref{inv:hierarchy:2} are satisfied. At this point, we check if the existing coloring $\chi$ is proper. The coloring $\chi$ can become invalid only if the edge $(u, v)$ is getting inserted and $\chi(u) = \chi(v)$. In this event, we find the endpoint $x \in \{u, v\}$ that was recolored last, i.e., the one with the larger $\tau_x$. Without any loss of generality, let this endpoint be $v$. We now change the color of $v$ by calling the subroutine RECOLOR$(v)$. At the end of the call to this subroutine, we know for sure that the coloring is proper. Thus, we can now apply the inductive hypothesis for the next insertion or deletion of an edge.

\begin{figure}[h!]
	\centerline{\framebox{
			\begin{minipage}{5in}
				On INSERT/DELETE($(u,v)$):
				\begin{itemize}
					\item[] MAINTAIN-HP.   \qquad // {\em See Figure~\ref{fig:hierarchy}. }
					\item[] In case $(u,v)$ is inserted and $\chi(u) = \chi(v)$:
					\begin{itemize}
						\item[] Suppose that $\tau_v > \tau_u$. \ \ \ // {\em  This notation is defined in Section~\ref{sec:randomized:data:structures}.}
						\item[]  RECOLOR($v$).
					\end{itemize}
				\end{itemize}
			\end{minipage}
	}}
	\caption{\label{fig:recolor:main} Dynamic algorithm to maintain $\Delta + 1$ vertex  coloring  }
\end{figure}

We first bound the amortized time spent on all the calls to the RECOLOR subroutine in Lemma~\ref{lm:bounding-recoloring}. From Theorem~\ref{th:hierarchy:maintain} and Lemma~\ref{lm:bounding-recoloring}, we get the main result of this section, which is stated  in Theorem~\ref{th:randomized:main}.
\begin{lemma}
	\label{lm:bounding-recoloring}
	Consider a sequence of $T$ edge insertions/deletions starting from an empty graph $G = (V, E)$.  Let $T_{R}$ and $T_{HP}$ respectively denote the total time spent on all the calls to the  RECOLOR and {\sc MAINTAIN-HP} subroutines during these $T$ edge insertions/deletions. Then  $\mathbf{E}[T_{R}] \leq O(T) + O(T_{HP})$.
\end{lemma}
\begin{proof}
	Since edge deletions don't lead to recoloring, we need to bother only with edge insertions.
	Consider  the scenario where an edge $(u,v)$ is being inserted into the graph at time $\tau$. Without any loss of generality, assume that $\tau_v > \tau_u$. Recall that these are the last times before $\tau$ when $v$ and $u$ were recolored. Suppose that the vertex $v$ is at  level $i$ immediately after we have updated the hierarchical partition following the insertion of the edge at time $\tau$. Thus, if $\chi(u) = \chi(v)$ at this point in time, then the subroutine RECOLOR($v$) will be called to change the color of the endpoint $v$.  On the other hand, if $\chi(u) \neq \chi(v)$ at this point in time, then no vertex will be recolored. Furthermore, suppose that the vertex $v$ was at level $j$  during the call to RECOLOR$(v)$ at time $\tau_v$. The analysis is done via three cases. \smallskip
	
	\noindent
	{\em Case 1: $i > j$.} In this case,  at some point in time during the interval $[\tau_v, \tau]$, the subroutine MAINTAIN-HP raised the level of the vertex $v$ to  $i$.
	 Corollary~\ref{cor:lm:hierarchy:time} implies that this takes  $\Omega(\beta^i)$ time. On the other hand, even if the subroutine RECOLOR$(v)$ is called at time $\tau$, by Lemma~\ref{lm:fix:time}  it takes $O(\beta^i)$ time to implement that call. So the total time spent on all such calls to the RECOLOR subroutine is at most  $O(T_{HP})$.\smallskip
	
	\noindent
	{\em Case 2: $4 < i \leq j$.} In this case, we use the fact that the vertex $v$ picks a random color at time $\tau_v$.
	In particular, by Lemma~\ref{lm:fix:time} the expected time spent on recoloring the vertex $v$ at time $\tau$ is at most $O(\beta^i)\cdot \Pr[\E_{\tau}]$, where $\E_{\tau}$ is the event that $\chi(u) = \chi(v)$  just before the insertion   at time  $\tau$. 
	We wish to bound this probability $\Pr[\E_{\tau}]$, which is evaluated over the past random choices of the algorithm {\em which the adversary fixing the order of edge insertions is oblivious to}.\footnote{In case 1, we used the trivial upper bound $\Pr[\E_{\tau}] \leq 1$.} We do this by using the principle of deferred decision.
	
	Let $c_{u}$ and $c_v$ respectively denote  the colors  assigned to the vertices   $u$  and $v$ during the calls to the subroutines RECOLOR$(u)$ and RECOLOR$(v)$ at times $\tau_u$ and $\tau_v$. Note that the event $\E_{\tau}$ occurs iff $c_u = c_v$. Condition on all the random choices made by the algorithm till just before the time $\tau_v$. Since $\tau_u < \tau_v$, this fixes the color $c_u$. At time $\tau_v$, the vertex $v$ picks the color $c_v$ uniformly at random from the subset of colors $B_v \cup U_v$. Let $\lambda$ denote the size of this subset $B_v \cup U_v$ at time $\tau_v$. Clearly,  the event $c_v = c_u$ occurs with probability $1/\lambda$, i.e., we have $\Pr[\E_{\tau}] = 1/\lambda$. It now remains to lower bound $\lambda$. Since $4 < i \leq j$, Claim~\ref{cl:blank-unique-size} and Invariant~\ref{inv:hierarchy:1} imply that when the vertex $v$ gets recolored at time $\tau_v$, we have: $\lambda = |B_v \cup U_v| \geq 1+|\N_v(4, j-1)|/2 \geq 1+ \beta^{j-5}/2 = \Omega(\beta^{j-5}) = \Omega(\beta^{i-5})$.  
To summarize, the expected time spent on the possible call to RECOLOR$(v)$ at time $\tau$ is at most $O(\beta^i) \cdot \Pr[\E_{\tau}] = O(\beta^i) \cdot (1/\lambda) = O(\beta^5) = O(1)$. Hence, the total time spent on all such calls to the RECOLOR subroutine is at most $O(T)$. \smallskip
	
	\noindent 
	{\em Case 3. $i = 4$.} Even if  RECOLOR$(v)$ is called at time $\tau$, by Lemma~\ref{lm:fix:time} at most $O(\beta^4) = O(1)$ time is spent on that call. So the total time spent on all such calls to the RECOLOR subroutine is $O(T)$. 
\end{proof}

\medskip
\noindent {\bf Proof of Theorem~\ref{th:randomized:main}.} The theorem holds since $\mathbf{E}[T_R] + T_{HP} \leq O(T) + O(T_{HP})   \leq O(T) + O(T  \log \Delta) = O(T \log \Delta)$. The first and the second inequalities respectively  follow from Lemma~\ref{lm:bounding-recoloring} and  Theorem~\ref{th:hierarchy:maintain}.

\section{A Deterministic Dynamic Algorithm for $(1+o(1))\Delta$ Vertex Coloring}
\label{sec:det:coloring}

Let $G = (V, E)$ denote the input graph that is changing dynamically, and let $\Delta$ be an upper bound on the maximum degree of any vertex in $G$. For now we assume that the value of $\Delta$ does not change with time. In Section~\ref{sec:general}, we explain how to relax this assumption.  Our main result is stated in the theorem below.

\begin{theorem}
	\label{th:analysis}
	We can maintain a $(1+o(1))\Delta$ vertex  coloring in a dynamic graph deterministically in $O(\lg^{4+o(1)} \Delta/\lg\lg^2 \Delta)$ amortized update time.
\end{theorem}

\subsection{Notations and preliminaries.}

Throughout Section~\ref{sec:det:coloring}, we  define three parameters $\eta$, $L$, $\lambda$ as follows.
\begin{eqnarray}
\label{eq:L}
\eta = e^{16/\lg \lg \Delta}, L = \left\lfloor \frac{\lg(\eta\Delta)}{\lg\lg\Delta}\right\rfloor  
{\rm   \ and \ }
\lambda = \left\lceil 2^{ \frac{\lg(\eta\Delta)}{L}}\right\rceil. 
\end{eqnarray}

We will use $\lambda^L$ colors. From~(\ref{eq:L}) and Lemma~\ref{lm:fix:1}, it follows that $\lambda^L  \leq \eta \Delta = (1+o(1))  \Delta$ when $\Delta = \omega(1)$. In Lemma~\ref{lm:fix:1}, we establish a couple of useful bounds on the parameters $\eta, L$ and $\lambda$.

\begin{lemma}
\label{lm:fix:1}
We have:
\begin{enumerate}
\item  $\lg \Delta \leq \lambda \leq 2 \lg^{1+o(1)} \Delta$, and
\item  $\lambda^L \le \eta\Delta \le (\lambda+1)^L$.
 \end{enumerate}
 \end{lemma}

\begin{proof}
From~(\ref{eq:L}) we infer that: 
\begin{eqnarray*}
\lambda \geq 2^{\frac{\lg (\eta \Delta)}{L}} \geq 2^{\frac{\lg (\eta \Delta)}{\lg (\eta \Delta)/\lg \lg \Delta}}  = \lg \Delta.
\end{eqnarray*}
From~(\ref{eq:L}) we also infer that:
\begin{eqnarray*}
\lambda & \leq & 2 \cdot 2^{\frac{\lg (\eta \Delta)}{L}} \\
& \leq & 2 \cdot 2^{\frac{\lg (\eta \Delta)}{ \lg (\eta \Delta)/\lg \lg \Delta - 1}}  \\ 
& = & 2 \cdot 2^{\frac{\lg (\eta \Delta)}{\lg (\eta \Delta) - \lg \lg \Delta} \cdot \lg \lg \Delta} \\
& = & 2 \cdot 2^{(1+o(1)) \cdot \lg \lg \Delta} = 2  \lg^{1+o(1)} \Delta.
\end{eqnarray*}
The proves part (1) of the lemma. Next, note that:
$$\lambda^L \le \left(2^{\frac{\lg(\eta\Delta)}{L}}\right)^L = \eta\Delta \le \left\lceil 2^{\frac{\lg(\eta\Delta)}{L}}\right\rceil^L = (\lambda+1)^L.$$
This proves part (2) of the lemma.
\end{proof}

We let  $\C = \{ 1, \ldots, \lambda^L \}$ denote the  palette of all colors. Note that  $|\C| = \lambda^L = (1+o(1)) \Delta$. Indeed, we view the colors available to us as $L$-tupled vectors where each coordinate takes one of the values from $\{1,\ldots,\lambda\}$.
In particular, the color  assigned to any vertex $v$ is denotes as $\chi(v) = (\chi_1(v),\ldots,\chi_L(v))$, where  $\chi_i(v)\in [\lambda]$ for each $i \in [L]$.
Given such a coloring $\chi : V \rightarrow \C$, for every index $i \in [L]$ we define $\chi^*_i(v) := (\chi_1(v), \ldots, \chi_i(v))$ to be the $i$-tuple denoting the first $i$ coordinates of $\chi(v)$. For notational convenience, we  define $\chi^*_0(v) := \bot$ for all $v$. 
For all $i \in [L]$ and $\alpha \in [\lambda]$, we let $\chi^*_{i \to \alpha}(v) = (\chi_1(v), \ldots, \chi_{i-1}(v), \alpha)$ denote the $i$-tuple whose first $(i-1)$ coordinates are the same as that of $\chi$ but whose $i^{th}$ coordinate is  $\alpha$. 

For all $i \in [L]$ and $\alpha \in [\lambda]$, we define the subsets $N_i^*(v) = \{ u \in V : (u, v) \in E {\rm  \ and \ } \chi_i^*(u) = \chi_i^*(v) \}$ and $N^*_{i \to \alpha}(v) = \{ u \in V : (u, v) \in E {\rm  \ and \ } \chi_i^*(u) = \chi^*_{i \to \alpha}(v) \}$. In other words, the set $N_i^*(v)$ consists of all the neighbors of a vertex $v \in V$ whose colors have the same first $i$ coordinates  as  the color of $v$. On the other hand, the set $N^*_{i \to \alpha}(v)$ denotes  the status of the set $N^*_{i}$  {\em in the event} that the vertex $v$ decides to change the $i$th coordinate of its color to $\alpha$. In particular, if $\chi_i(v) \neq \alpha$, then $N_i^*(v) \cap N_{i\to \alpha}^*(v) = \emptyset$. Going over all possible choices of $\chi_i(v)$ we get the following
\[
N^*_{i-1}(v) =  \bigcup_{\alpha\in [\lambda]} N^*_{i\to \alpha} (v)
\]
Also note that $N^*_0(v)$ is the full neighborhood of $v$.
Define $D_i^*(v) = |N_{i}^*(v)|$ and $D^*_{i \to \alpha}(v) = |N^*_{i \to \alpha}(v)|$.  The above observations is encapsulated in the following corollary.
\begin{corollary}
	\label{cor:notation}
	For every vertex $v \in V$, and every index $j \in \{1, \ldots, L\}$, the set $N^*_{j-1}(v)$ is partitioned into the subsets $N^*_{j \rightarrow \alpha}(v)$ for $\alpha \in \{1, \ldots, \lambda\}$. In particular, we have: $D^*_{j-1}(v) = \sum_{\alpha \in [\lambda]} D^*_{j \rightarrow \alpha}(v)$.
\end{corollary}

We  maintain the following invariant.

\begin{invariant}
	\label{inv:main}
	For all $v \in V$,  $i \in [0, L]$, we have $D^*_i(v) \leq (\Delta/\lambda^i) \cdot f(i)$, where $f(i) = ( (\lambda+1)/(\lambda -1))^i$. 
\end{invariant}

For every $j$, we have $f(j) > 1$ and $f(j-1) \leq f(j)$. We now give a brief intuitive explanation for the above invariant.  Associate a rooted $\lambda$-ary  tree $T_v$ of depth $L$ with every vertex $v \in V$. We shall refer to the vertices   of this tree $T_v$ as meta-vertices, to distinguish them from the vertices   of the input graph $G = (V, E)$. The total number of leaves in this tree is $\lambda^L$, which is the same as the total number available colors. Thus, we can ensure that each root to leaf path in this tree corresponds to a color in a natural way, and any internal meta-vertex at depth $i$ corresponds to the  $i$th coordinate of a color. The quantity $D^*_i(v)$ can now be interpreted as follows. Consider the meta-vertex (say) $x_i$ at depth $i$ on the unique root to leaf path corresponding to the color of $v$. Let $\mu_i$ denote the number of all neighbors of $v$ in $G$ such that this meta-vertex $x_i$ also belongs to the root to leaf paths for their corresponding colors. Then we have $D^*_i(v) = \mu_i$. Note that if $i = 0$, then the meta-vertex $x_i$ is the root of the tree, which is at depth zero. It follows that if $i = 0$, then $\mu_i$ equals the degree of the vertex $v$ in the input graph $G$, which is at most $\Delta$. Thus, we  have $\mu_0 \leq \Delta$. Now, let $y_1, \ldots, y_{\lambda}$ denote the children of the root $x_0$ in this $\lambda$-ary tree $T_v$. A simple counting argument implies that there exists an index $j \in \{1, \ldots, \lambda\}$ with the following property: At most  a $1/\lambda$ fraction of the neighbors of $v$ in $G$ have colors whose corresponding root to leaf paths contain the meta-vertex $y_j$. Thus, if it were the case that the root to leaf path corresponding to the color of $v$ also passes through such a meta-vertex $y_j$, then we would have  $D^*_i(v) = \mu_i \leq \Delta/\lambda$ for $i = 1$.  Invariant~\ref{inv:main}, on the other hand, gives a {\em slack} of $f(1)$ and requires that $D^*_i(v) \leq (\Delta/\lambda) \cdot f(1)$ for $i = 1$. We can interpret the invariant in this fashion for every subsequent index $i \in \{2, \ldots, L\}$ by iteratively applying the same principle. The reader might find it helpful to keep this interpretation in mind while going through the formal description of the algorithm and its analysis.

\begin{lemma}\label{lem:invariant-implies-coloring}
	If Invariant~\ref{inv:main} holds then $\chi$ is a proper vertex  coloring.
\end{lemma}
\begin{proof}
	Claim~\ref{cl:invariant:boundary}  implies that for $i = L$, the invariant reduces to: $D^*_L(v) < 1$. Since $D^*_L(v)$ is a nonnegative integer, we get $D^*_L(v) = 0$.  Since $D^*_L(v)$ is the number of neighbors of $v$ who are assigned the color $\chi(v)$,  no two adjacent vertices   can get the same color. The invariant thus ensures a proper vertex  coloring.
	
	\begin{claim}
		\label{cl:invariant:boundary}
		We have: $(\Delta/\lambda^L) \cdot f(L) < 1$.
	\end{claim}

	In order to prove Claim~\ref{cl:invariant:boundary}, we derive that:		
\begin{eqnarray}
   & &  (\Delta/\lambda^L)f(L)  \nonumber \\
      & \le & (((\lambda+1)^L/\eta) / \lambda^L) \cdot f(L) \label{eq:new:derive:1} \\
    & = & (1/\eta) \cdot (1 + 1/\lambda)^L \cdot (1+2/(\lambda-1))^L \nonumber \\
   &  \leq & (1/\eta) \cdot (1 + 1/\lambda)^L \cdot (1+4/\lambda)^L  \label{eq:new:derive:2} \\
   & \le & (1/\eta) \cdot (1 + 7/\lambda)^L \label{eq:new:derive:3} \\
   & \le & (1/\eta) \cdot e^{7L/\lambda} \nonumber \\
   & \leq & (1/\eta) \cdot e^{7 \lg (\eta \Delta)/(\lambda \lg \lg \Delta)} \label{eq:new:derive:4} \\
   & \leq & (1/\eta) \cdot e^{7 \lg (\Delta^2)/ (\lg \Delta \lg \lg \Delta)} \label{eq:new:derive:5} \\
   & \leq & (1/\eta) \cdot e^{14/\lg \lg \Delta} \nonumber \\
   & <  & 1. \label{eq:new:derive:6}
 \end{eqnarray}
Step~(\ref{eq:new:derive:1}) follows from part (2) of Lemma~\ref{lm:fix:1}. Steps~(\ref{eq:new:derive:2}),~(\ref{eq:new:derive:3}) hold as long as $\lambda \geq  2$.\footnote{When $\lambda < 2$, we have $\Delta= O(1)$ and so we can trivially maintain a $(\Delta+1)$-vertex coloring in $O(\Delta) = O(1)$ update time.} Step~(\ref{eq:new:derive:4}) follows from~(\ref{eq:L}). Step~(\ref{eq:new:derive:5}) follows from~(\ref{eq:L}) and part (1) of Lemma~(\ref{lm:fix:1}). Step~(\ref{eq:new:derive:6}) follows from~(\ref{eq:L}).
\end{proof}

\paragraph{Data structures.} For every vertex $v \in V$, our dynamic algorithm maintains the following data structures.
\begin{itemize}
	\item For all $i \in [0, L]$, the set $N_i^*(v)$ as a doubly linked list and the counter $D_i^*(v)$.  It will be the responsibility of the neighbors of $v$ to update the list $N_i^*(v)$ when they change their own colors. Using appropriate pointers, we will ensure that any given node $x$ can be inserted into or deleted from any given list $N_i^*(y)$ in $O(1)$ time. Note that each vertex figures out if it satisfies the Invariant~\ref{inv:main} or not. 
	\item The color $\chi(v) = (\chi_1(v), \ldots, \chi_L(v))$ assigned to the vertex $v$. 
\end{itemize}

\subsection{The algorithm.}
\label{sub:sec:algo}

Initially, since the edge-set of the input graph  is empty,  we can assign each vertex an arbitrary color. For concreteness, we set $\chi_i(v) = 1$ for all $i\in [L]$ and $v \in V$. Since there are no edges, $D^*_i(v) = 0$ for all $i,v$ and so Invariant~\ref{inv:main} vacuously holds at this point. 
We show how to ensure that the invariant continues to remain satisfied even after any edge insertion or deletion, and bound the amortized update time.

\paragraph{Deletion of an edge.} Suppose that an edge $(u, v)$ gets deleted from $G$. No vertex changes its color due to this deletion. We  only need to  update the relevant data structures. Without any loss of generality, suppose that $i \in [0, L]$ is the largest index for which we have $\chi_i^*(v) = \chi_i^*(u)$. Then for every vertex $x \in \{u, v\}$ and every $j \in [0, i]$, we delete the vertex $y \in \{ u, v \} \setminus \{x\}$ from the set $N_i^*(x)$ and decrement the value of the counter $D_j^*(x)$ by one.  This takes $O(L)$ time. To summarize, deletion of an edge can be handled in $O(L)$ worst case update time. If Invariant~\ref{inv:main} was satisfied just before the edge deletion, then the invariant continues to remain satisfied after the edge deletion since the LHS of the invariant can only decrease.

\paragraph{Insertion of an edge.} Suppose that an edge $(u, v)$ gets inserted into the graph. We first update the relevant data structures as follows. Let $i \in [0, L]$ be the largest index such that $\chi_i^*(u) = \chi_i^*(v)$. Note that if $i=L$ then the colors of $u$ and $v$ are the same.
For every vertex $x \in \{u, v\}$ and every $j \in [0, i]$, we insert the vertex $y \in \{u,v\} \setminus \{x\}$ into the set $N_j^*(x)$ and increment the value of the counter $D_j^*(x)$ by one. This takes $O(L)$ time in the worst case. Next, we focus on ensuring that Invariant~\ref{inv:main} continues to hold. Towards this end, we execute the subroutine described in Figure~\ref{fig:while}. Lemma~\ref{lm:correct} implies the correctness of the algorithm. In Lemma~\ref{lm:run}, we upper bound the time spent on a given iteration of the {\sc While} loop in Figure~\ref{fig:while}.

\begin{figure}[h!]
	\centerline{\framebox{
			\begin{minipage}{5.5in}
				\begin{tabbing}
					1. \=  {\sc While} there is some vertex $x \in V$  that violates Invariant~\ref{inv:main}:  \\
					2. \> \ \ \ \ \ \ \ \    \=  Let $k \in [L]$ be the smallest index  such that $D^*_k(x) > (\Delta/\lambda^k) \cdot f(k)$. \\
					3. \> \> {\sc For} $j = k$ to $L$: \\
					4. \> \> \ \ \ \ \ \ \ \ \  \=  Find an $\alpha \in [\lambda]$ that minimizes $D^*_{j \to \alpha}(x)$. \\
					5. \> \> \>  Set $\chi_j(x) = \alpha$ and update the relevant data structures.  \qquad  // See Lemma~\ref{lm:run}.
				\end{tabbing} 
			\end{minipage}
	}}
	\caption{\label{fig:while} Ensuring Invariant~\ref{inv:main} after an edge insertion. }
\end{figure}

\begin{lemma}
	\label{lm:correct}
	Consider an iteration of the {\sc While} loop in Figure~\ref{fig:while} for a vertex $x \in V$. At the end of this iteration, the vertex $x$ satisfies Invariant~\ref{inv:main}.
\end{lemma}

\begin{proof}
	Consider any iteration of the {\sc For} loop in steps 3 -- 5. By induction hypothesis, suppose that just before this iteration we have $D^*_{i}(x) \leq (\Delta/\lambda^{i}) \cdot f(i)$ for all $i \in [0, j-1]$. Due to step 2, the induction hypothesis  holds just before the first iteration of the {\sc For} loop, when we have $j = k$. Since $D^*_{j-1}(x) = \sum_{\alpha \in [\lambda]} D^*_{j \to \alpha}(x)$ by Corollary~\ref{cor:notation}, the $\alpha$ that minimizes $D^*_{j \to \alpha}(x)$ satisfies: $D^*_{j \to \alpha}(x) \leq (1/\lambda) \cdot D^*_{j-1}(x) \leq (\Delta/\lambda^j) \cdot f(j-1) \leq (\Delta/\lambda^j) \cdot f(j)$. Accordingly, after executing step 6 we get: $D^*_j(x) \leq (\Delta/\lambda^j) \cdot f(j)$. Thus, the induction hypothesis remains valid for the next iteration of the {\sc For} loop. At the end of the {\sc For} loop, we get $D^*_j(x) \leq (\Delta/\lambda^j) \cdot f(j)$ for all $j \in [0, L]$, and hence the vertex $x$ satisfies Invariant~\ref{inv:main}. 
\end{proof}

\begin{lemma}
	\label{lm:run}
	It takes $O(\lambda + L \cdot \frac{\Delta}{\lambda^{k-1}} \cdot f(k-1))$ time for one iteration of the {\sc While} loop in Figure~\ref{fig:while}, where $k$ is defined as per Step 2 in Figure~\ref{fig:while}.
	\end{lemma}

\begin{proof}
	Consider any iteration of the {\sc While} loop that changes the color of a vertex $x \in V$. Since we store the value of $D^*_j(x)$ for every $j \in [0, L]$,  it takes $O(L)$ time to find the index $k \in [0, L]$ as defined in step 2 of Figure~\ref{fig:while}. Next, for every index $j \in [k, L]$ and every vertex $u \in N^*_j(x)$, we set $N^*_j(u) = N^*_j(u) \setminus \{x\}$ and $D^*_j(u) = D^*_j(u) - 1$. Since the vertex $x$ is going to change the $k$th coordinate of its color, the previous step is necessary to ensure that no vertex $u$ mistakenly continues to include $x$ in the set $N^*_j(u)$ for $j \in [k, L]$. This takes $O(\sum_{j=k}^L |N^*_j(x)|)$ time. Since $N^*_j(x) \subseteq N^*_{j-1}(x)$ for all $j \in [k, L]$, the time taken is actually $O((L-k-1) \cdot |N^*_{k-1}(x)|) = O(L \cdot (\Delta/\lambda^{k-1}) \cdot f(k-1))$. The last equality follows from step 2 in Figure~\ref{fig:while}. At this point, we also set $N^*_j(x) = \emptyset$ and $D^*_j(x) = 0$ for all $j \in [k, L]$. We shall rebuild the sets $N^*_j(x)$ during the  {\sc For} loop in steps 3 - 5. Applying a similar argument as before, we conclude that this also takes $O(L \cdot (\Delta/\lambda^{k-1}) \cdot f(k-1))$ time. It now remains to bound the time spent on the {\sc For} loop  in Figure~\ref{fig:while}. 
	
	Before  the first iteration of the {\sc For} loop,  we set a counter $Z_{\tau} = 0$ for all $\tau \in [\lambda]$. This takes $\Theta(\lambda)$ time.
	
	Now, consider any iteration of the {\sc For} loop. By inductive hypothesis, suppose that for every index $i \in [0, j-1]$ and every vertex $v \in V$, the list $N^*_i(v)$ is now consistent with the changes we have made to the color of $x$ in the earlier iterations of the {\sc For} loop.  Our first goal is to  compute the index $\alpha \in [\lambda]$, which we do by scanning the vertices   $u \in N^*_{j-1}(x)$. Just before the scan begins, we have $Z_{\tau} = 0$ for every $\tau \in [\lambda]$. While considering any vertex $u \in N^*_{j-1}(x)$ during the scan, we set $Z_{c_j(u)} = Z_{c_j(u)}+1$. At the end of the scan, we return the index $\alpha \in [\lambda]$ that minimizes the value of $Z_{\alpha}$.  Subsequently, we again scan through the vertices in $N^*_{j-1}(x)$ and set $Z_{c_j(u)} = Z_{c_j(u)} - 1$ for every vertex $u \in N^*_{j-1}(x)$. The second scan ensures that we again have $Z_{\tau} = 0$ for every $\tau \in [\lambda]$. Thus,  it takes $O(|N^*_{j-1}(x)|)$  time  to find the index $\alpha$ and perform the two scans of the list $N^*_{j-1}(x)$. Next, for every vertex $u \in N^*_{j-1}(x)$ with $c_j(u) = \alpha$,  we set $N^*_j(u) = N^*_j(u) \cup \{x\}$, $N^*_j(x) = N^*_j(x) \cup \{u\}$, $D^*_j(u) = D^*_j(u) + 1$ and $D^*_j(x) = D^*_j(x) + 1$. It takes  $O(|N^*_{j-1}(x)|)$ time to implement this step. Now, we set $c_j(x) = \alpha$ in $O(1)$ time. At this stage, we have updated all the relevant data structures for the index $j$ and  changed the $j$th coordinate of the color of $x$. This concludes the concerned iteration of the {\sc For} loop. Since $N^*_{j-1}(x) \subseteq N^*_{k-1}(x)$, the total time spent on this iteration is  $O(|N^*_{j-1}(x)|) = O(|N^*_{k-1}(x)|) = O((\Delta/\lambda^{k-1}) \cdot f(k-1))$, as per step 2 in Figure~\ref{fig:while}.
	
	Since the {\sc For} loop runs for at most $L$ iterations, the total time spent on the {\sc For} loop is at most $O(\lambda + L  \cdot \frac{\Delta}{\lambda^{k-1}} \cdot  f(k-1))$. Combining this  with the discussion in the first paragraph of the proof of this lemma, we infer that the total time spent on one  iteration of the {\sc While} loop is also $O(\lambda + L \cdot \frac{\Delta}{\lambda^{k-1}} \cdot f(k-1))$.
\end{proof}

\subsection{Bounding the amortized update time.}

In this section, we prove Theorem~\ref{th:analysis} by bounding the amortized update time of the algorithm described in Section~\ref{sub:sec:algo}. Recall that handling the deletion of an edge takes $O(L)$ time in the worst case. Furthermore, ignoring the time spent on the {\sc While} loop in Figure~\ref{fig:while}, handling the insertion of an edge also takes $O(L)$ time in the worst case. From Lemma~\ref{lm:fix:1} we have $L = O(\lg \Delta/\lg\lg\Delta)$. Thus, it remains to bound the time spent on the {\sc While} loop in Figure~\ref{fig:while}. We focus on this task for the remainder of this section. 

The main idea is the following: by Lemma~\ref{lm:run} the {\sc While} loop processing vertex  $x$ takes a {\em long} time when $k$ is {\em small} which in turn implies $D^*_k(x)$ is {\em large}. However, in the {\sc For} loop we choose the colors which minimize precisely the values of $D^*_{j\to \alpha}(x)$. Therefore these quantities cannot be large too often.

Consider any iteration of the {\sc While} loop in Figure~\ref{fig:while}, which changes the color of a vertex $x \in V$. Let $S^+_x$ (resp. $S^-_x$) be the set of all ordered pairs $(i, v)$ such that the value of $D^*_i(v)$ increases (resp. decreases)  due to this iteration. The following lemma precisely bounds the increases and decreases of these $D^*$ values.

\begin{lemma}
	\label{lm:analyze:1}  During  any single iteration of the {\sc While} loop in Figure~\ref{fig:while}, we have:
	\begin{eqnarray*}
		|S^-_x|  >  (\Delta/\lambda^k) \cdot f(k)  {\rm \  and \ }  
		|S^+_x|  <  (\Delta/\lambda^k) \cdot f(k-1) \cdot \lambda/(\lambda-1).
	\end{eqnarray*}
\end{lemma}

\begin{proof}
	Throughout the proof, we let $t^-$ and $t^+$ respectively denote the time-instant just before and just after the concerned iteration of the {\sc While} loop.
	Let $k \in [L]$ be the smallest index such that $D^*_k(x) > (\Delta/\lambda^k) \cdot f(k)$ at time $t^-$.  For every index $j \in [0, L]$ and every vertex $v \in V$, let $N^*_j(v, t^-)$ and  $N^*_j(v, t^+)$ respectively denote the set of vertices   in $N^*_j(v)$ at time $t^-$ and at time $t^+$.

	Consider any vertex $u \in N^*_k(x, t^-)$. At time $t^-$, we had $\chi^*_k(u) = \chi^*_k(x)$.  The concerned iteration of the {\sc While} loop changes the $k$th coordinate of the color of $x$, but the vertex $u$ does not change its color during this iteration. Thus, we have $\chi^*_k(u) \neq \chi^*_k(x)$ at time $t^+$. We therefore infer that $x \in N^*_k(u, t^-)$  and $x \notin N^*_k(u, t^+)$. Hence, for every vertex $u \in N^*_k(x, t^-)$, the value of $D^*_k(u)$ drops by one due to the concerned iteration of the {\sc While} loop. It follows that $\{ (k, u) | u \in N^*_k(x, t^-) \} \subseteq S^-_x$, and we get: $|S^-_x| \geq |N^*_k(x, t^-)| > (\Delta/\lambda^k) \cdot f(k)$. The last inequality follows from step 2 in Figure~\ref{fig:while}.
	This gives us the desired lower bound on the size of the set $S^-_x$.  It now remains to upper bound the size of the set $S^+_x$. 
	
	Consider any ordered pair  $(j, u) \in S^+_x$. Since  the concerned iteration of the {\sc While} loop increases the value of $D^*_j(u)$ and does not change the color of any vertex other than $x$, we infer that: 
	\begin{equation}
	\label{eq:lm:analyze:1:1}
	x \notin N^*_j(u, t^-) {\rm  \ and \ } x \in N^*_j(u, t^+).
	\end{equation}
	The concerned iteration of the {\sc While} loop does not change the $i$th coordinate of the color of $x$ for any $i < k$. Thus, the $i$-tuple $\chi^*_i(x)$ does not change from time $t^-$ to time $t^+$. Furthermore, the vertex $u$  does not change its color during the time-interval $[t^-, t^+]$. It follows that if $x \notin N^*_i(u, t^-)$ for some $i < k$, then we also have $x \notin N^*_i(u, t^+)$. From~(\ref{eq:lm:analyze:1:1}) we therefore get $j \in [k, L]$. Next, note that since $x \in N^*_j(u, t^+)$, we have $\chi^*_j(x) = \chi^*_j(u)$ at time $t^+$, and accordingly we also have $u \in N^*_j(x, t^+)$. To summarize, if an ordered pair $(j, u)$ belongs to the set $S^+_x$, then we must have $j \in [k, L]$ and $u \in N^*_j(x, t^+)$.  Thus, we get:
	\begin{equation}
	\label{eq:1}
	|S^+_x| \leq \sum_{j = k}^L |N^*_j(x, t^+)|
	\end{equation}
	Note that step 4 in Figure~\ref{fig:while} picks an $\alpha \in [\lambda]$ that minimizes $D^*_{j \to \alpha}(x)$. This  gives us the following guarantee.
	\begin{eqnarray}
	\label{eq:20}
	|N^*_k(x, t^+)| & \leq &  |N^*_{k-1}(x, t^-)|/\lambda. \\
	\label{eq:2}
	|N^*_j(x, t^+)| & \leq &  |N^*_{j-1}(x, t^+)|/\lambda  {\rm  \ for \ all \ } j \in [k+1, L].
	\end{eqnarray}
	Using~(\ref{eq:1}),~(\ref{eq:20}) and~(\ref{eq:2}), we can upper bound $|S^+_x|$ by the sum of a geometric series, and get:
	\begin{eqnarray}
	|S^+_x| & \leq & \sum_{j=k}^L |N^*_j(x, t^+)|   \nonumber \\
	& \leq &  |N^*_{k-1}(x, t^-)| \cdot \left(\frac{1}{\lambda} +  \cdots + \frac{1}{\lambda^{L-k+1}} \right) \nonumber \\
	& \leq & |N^*_{k-1}(x, t^-)| \cdot \frac{1}{(\lambda - 1)}  \nonumber \\
	& \leq & \frac{\Delta}{\lambda^{k-1}} \cdot f(k-1)\cdot \frac{1}{(\lambda -1)}  \nonumber
	\end{eqnarray}
	The last inequality holds due to step 2 in Figure~\ref{fig:while}.  This gives us the desired upper bound on $|S^+_x|$.
\end{proof}

We now use a potential function based argument to prove Theorem~\ref{th:analysis}, where the potential associated with the graph at any point in time is given by $\Phi = \sum_{v \in V, i \in [0, L]} D_i^*(v)$. Note that the potential $\Phi$ is always nonnegative. The bound on the amortized update now follows from the three observations stated below.

\begin{observation}
	\label{ob:1} 
	Due to the insertion or deletion of an edge in $G$, the potential $\Phi$ changes by at most $O(L)$. 
\end{observation}

\begin{proof}
	Consider the insertion or deletion of an edge $(u, v)$ in the input graph $G = (V, E)$. For all $x \in V \setminus \{u, v\}$ and  $i \in [0, L]$, the value of $D_i^*(x)$ remains unchanged. Furthermore, for all $x \in \{u, v\}$ and  $i \in [0, L]$, the value of $D_i^*(x)$ changes by at most one. Hence, the potential changes by at most $2(L+1)$.
\end{proof}

\begin{observation}
	\label{ob:2}
	Due to one iteration of the {\sc While} loop in Figure~\ref{fig:while}, the potential $\Phi$ decreases by at least $\Gamma = (\Delta/\lambda^{k}) \cdot f(k-1)/(\lambda-1)$, where $k$ is defined as per Step 2 in Figure~\ref{fig:while}.
\end{observation}

\begin{proof}
	Note that the concerned iteration of the {\sc While} loop does not change the color of any vertex $v \neq x$. Thus, for every vertex $v \in V$ and every index $i \in [0, L]$, the value of $D^*_i(v)$ changes by at most one. As in Lemma~\ref{lm:analyze:1}, let $S^+_x$ (resp. $S^-_x$) denote the set of all ordered pairs $(i, v)$ such that the value of $D^*_i(v)$ increases (resp. decreases)  due to this iteration. We therefore infer that the net decrease in $\Phi$ is equal to $|S^-_x| - |S^+_x|$, and we have:
	\begin{eqnarray*}
	|S^-_x| - |S^+_x|  & \geq & (\Delta/\lambda^k) \cdot f(k) - (\Delta/\lambda^k) \cdot f(k-1) \cdot \lambda/(\lambda-1) \label{eq:300} \\
	& = &  (\Delta/\lambda^k) \cdot f(k-1) \cdot \left( (\lambda+1)/(\lambda-1) - \lambda/(\lambda-1)\right) \label{eq:301} \\
	& = & (\Delta/\lambda^k) \cdot f(k-1) /(\lambda-1). \nonumber
	\end{eqnarray*} The first step in the above derivation  follows from Lemma~\ref{lm:analyze:1}, and the second step follows from Invariant~\ref{inv:main}.
\end{proof}

\begin{observation}
	\label{ob:3}
	The time taken to implement one iteration of the {\sc While} loop in Figure~\ref{fig:while} is  $O((\lambda-1) \cdot \lambda \cdot L \cdot \Gamma)$, where $\Gamma$ is the net decrease in the potential due to the same iteration of the {\sc While} loop. 
\end{observation}

\begin{proof}
	This follows from Observation~\ref{ob:2} and Lemma~\ref{lm:run}.
\end{proof}

\noindent {\bf Proof of Theorem~\ref{th:analysis}.}
	In the beginning when the edge set is empty we have $\Phi = 0$.
	After $T$ edge insertions and deletions, the total time taken by the algorithm is $O(LT) + W$, where $W$ is the total time
	taken by the WHILE loops. 	
	From Observation~\ref{ob:3} we know that $W = O(\lambda^2 L\cdot \sum_t \Gamma_t)$, where $\Gamma_t$ is the decrease in the potential due to the $t$th WHILE loop. On the other hand, from Observation~\ref{ob:1} we get  $\sum_t \Gamma_t = O(TL)$. Putting it all together, we get that the total time taken by the algorithm is $O(\lambda^2 L^2T)$. This proves the theorem since $\lambda = O(\log^{1+o(1)} \Delta)$ and $L = O(\log \Delta/\log\log\Delta)$ as per~(\ref{eq:L}) and part (1) of Lemma~\ref{lm:fix:1}.

\section{A Deterministic Dynamic Algorithm for $(2\Delta - 1)$ Edge Coloring}
\label{sec:edge-coloring}

Let $G = (V, E)$ be the input graph that is changing dynamically, and let $\Delta$ be an upper bound on the maximum degree of any vertex in $G$. For now we assume that the value of $\Delta$ does not change with time. In Section~\ref{sec:general}, we explain how to relax this assumption. 
We present a simple, deterministic dynamic algorithm for maintaining a $2\Delta - 1$ edge coloring algorithm in $G$.

\paragraph{Data Structures.} 
For every vertex  $v \in V$, we maintain the following data structures.
\begin{enumerate}
\item An array $C_v$ of length $2\Delta-1$. Each entry in this array corresponds to a color. For each color $c$, the entry  $C_v[c]$ is either null or points to the unique edge incident on $v$ which is colored $c$. \item  A bit vector $A_v$ of  length $2\Delta-1$, where $A_v[c] = 0$ iff $C_v[c]$ is null, and $A_v[c]=1$ otherwise. 
\item A balanced binary search tree $T_v$ with $2\Delta-1$ leaves. We refer to the vertices in the tree as {\em meta-nodes}, to distinguish them from the vertices in the input graph $G$. We maintain a counter $val(T_v, x)$ at every meta-node $x$ in the tree $T_v$.  The value of this counter at the $c$th leaf of $T_v$ is given by $A_v[c]$. Furthermore, the  value of this counter  at any internal meta-node $x$ of $T_v$ is given by: $val(T_v, x) = \sum_{c : c {\rm  \ is \ a \ leaf \ in \ the \ subtree \ rooted \ at \ } x} val(T_v, c)$.
\end{enumerate}
We  use the notation $A_v[a:b]$ to denote $\sum_{a\leq c<b} A_v[c]$.

\paragraph{Initialization.} Initially when the graph $G = (V, E)$   is empty, $C_v$ is set to all nulls, $A_v$ is set to all $0$'s, and the counters   $val(T_v, x)$ are set to all $0$'s.

\paragraph{Coloring subroutine.} When an edge $e=(u,v)$ is inserted, we need to assign it a color in $\{1,2,\ldots,2\Delta-1\}$.
We do so using the following binary-search like procedure described in Figure~\ref{fig:edge}.

\begin{figure}[h!]
	\centerline{\framebox{
			\begin{minipage}{5.5in}
				\begin{tabbing}
					1. \  \=  Set $\ell=1$ and $r=2\Delta$. \\
					2. \> \=  {\sc While} $\ell<r-1$:  \qquad \qquad // \= {\em Invariant: $A_u[\ell:r] + A_v[\ell:r] < r-\ell$}\\
					3. \> \> \ \ \ \ \ \ \ \ \ \ \=  $z := \lceil(\ell+r)/2\rceil$\\
					4. \> \> \>  If $A_u[\ell:z] + A_v[\ell:z] < z - \ell$, {\sc Then} \\
					\> \> \> \ \ \ \ \ \ \ \= set $r=z$.\\
					5. \> \> \>  Else if $A_u[z:r] + A_v[z:r] < r - z$, {\sc Then} \\
					\> \> \> \> set $\ell = z$.\\
					6. \> \> Assign color $\chi(e) = \ell$ to edge $e$.\\
	 				7. \> \> Update the following data structures:   $C_u[\ell],C_v[\ell],A_u[\ell],A_v[\ell],T_u,T_v$.
				\end{tabbing} 
			\end{minipage}
	}}
	\caption{\label{fig:edge} COLOR($e$): Coloring an edge $e = (u, v)$ that has been inserted. }
\end{figure}

\begin{claim}
	\label{cl:proper-coloring}
	The subroutine COLOR($e$) returns a proper coloring of an edge $e$ in $O(\log \Delta)$ time.
\end{claim}
\begin{proof}
At the end of the {\sc While} loop we have $\ell=r-1$ since $z\leq r-1$.
Assume that the invariant stated in the while loop holds at the end; in that case we have $A_u[\ell]+A_v[\ell] < 1$ implying both are $0$.
This, in turn, means that there is no edge incident on either $u$ or $v$ with the color $\ell$. Thus, coloring $(u,v)$ with $\ell$ 
is proper. We now show that the invariant always hold. It holds in the beginning since at that point both $u$ and $v$ have degree $\le \Delta-1$ since $(u,v)$ is being added. This implies the total sum of $A_u$ and $A_v$ are $\leq 2\Delta-2$, which implies the invariant.
The while loop then makes sure that the invariant holds subsequently. 
For time analysis, note that there are $O(\log \Delta)$ iterations, and the values of $A_v[\ell:z]$'s are stored in the counters $val(T_v, x)$ at the internal meta-nodes  $x$  of the trees $T_v$, and these values can be obtained in $O(1)$ time. The data structures can also be maintained in $O(\log \Delta)$ time since in the tree $T_v$ the values of the meta-nodes   only in the path from $\ell$ to the root has to be increased by $1$.
\end{proof}
The full algorithm is this: initially the graph is empty. When an edge $e=(u,v)$ is added, we run COLOR($e$). When an edge $e=(u,v)$ is deleted and its color was $c$, we point $C_u[c]$ and $C_v[c]$ to null, set $A_u[c]=A_v[c]=0$ and then update $T_u$ and $T_v$ in $O(\log \Delta)$ time. Thus we have the following theorem.
\begin{theorem}
	In a dynamic graph with maximum degree $\Delta$, we can deterministically maintain a $(2\Delta-1)$-edge coloring in $O(\log \Delta)$ worst case update time.
\end{theorem}

\section{Extensions to the Case where $\Delta$ Changes with Time}
\label{sec:general}
For ease of exposition, in the whole paper we have maintained that $\Delta$ is a parameter known to the algorithm up front with the promise that maximum degree of the graph remains at most $\Delta$ at all times. In fact all our algorithms, with some work, can work with the changing $\Delta$ as well. That is, if $\Delta_t$ is the maximum degree of the graph after $t$ edge insertions and deletions, then in fact we have a randomized algorithm maintaining a $(\Delta_t+1)$ vertex  coloring, a deterministic algorithm maintaining a $(1+o(1))\Delta_t$ vertex  coloring, and a deterministic algorithm maintaining a $(2\Delta_t - 1)$ edge coloring. Our running times take a slight hit. For the first two algorithms the amortized running time is $O(\polylog \Delta)$ where $\Delta$ is the maximum degree seen so far (till time $t$); for the edge coloring the worst case update time becomes $O(\log \Delta_t)$. In the following three subsections we give a brief sketch of the differences in the algorithm and how the analysis needs to be modified. In all cases, we achieve this by making the requirement on the algorithm stronger.

\subsection{Randomized $(\Delta_t + 1)$ vertex  coloring.}

Let $D_v = |\N_v(4, L)|$ be the degree of a vertex $v \in V$ in the current input graph. To extend our dynamic algorithm in Section~\ref{sec:randomized:coloring} to the scenario where $\Delta$ changes with time, we  simply ensure that the following property holds. 
\begin{property}
	\label{prop:delta}
	Every vertex $v$ to have a color $\chi(v) \in \{1, \ldots,  D_v+1\}$.
\end{property}
\noindent To see why it is easy to ensure Property~\ref{prop:delta}, we only need to make the following two observations.
\begin{enumerate}
	\item The algorithm that maintains the hierarchical partition in Section~\ref{sec:hierarchy} is oblivious to the value of $\Delta$.
	\item For every vertex $v \in V$,  change the definition of the subset of colors $\C_v \subseteq \C$ as follows. Now, the subset $\C_v \subseteq \C$ consists of all the colors in $\{1, \ldots, D_v+1\}$ that are not assigned to any neighbor $u$ of $v$ with $\ell(u) \geq \ell(v)$. We can maintain these modified sets $\C_v$ by incurring only a $O(1)$ factor cost in the update time. Finally, note that even with this modified definition, Claim~\ref{cl:blank-unique-size} continues to hold. Hence, the RECOLOR subroutine in Figure~\ref{fig:recolor} in particular and our randomized dynamic coloring algorithm in general continue to remain valid.
\end{enumerate}

\subsection{Deterministic $(1+o(1))\Delta_t$ vertex  coloring.}
Whenever $\deg_t(v)$ is less than a (very) large constant, whenever we need to change its color we do the greedy step taking $\deg_t(v) = O(1)$ time to find a free color. Henceforth assume $\deg_t(v) = \omega(1)$.
Instead of having a  fixed $\lambda,L$ and $\eta$, 
for every vertex  $v$ we have separate parameters which depend on the degree of $v$ at time $t$.
Note these can be maintained at each insertion and deletion with $O(1)$ time per update.
\begin{eqnarray*}
	\eta_t(v) := e^{\frac{16}{\lg\lg \deg_t(v)}}, ~~ L_t(v) := \left\lfloor \frac{\lg(\eta_t(v)\cdot \deg_t(v))}{\lg\lg \deg_t(v)}\right\rfloor 
	\textrm{ and } \lambda_t(v):= \left\lceil2^{\frac{\lg(\eta_t(v)\cdot \deg_t(v))}{L_t(v)}} \right\rceil.
\end{eqnarray*}


At time $t$, each vertex  is assigned a color from $\{1,2,\ldots, \lambda_t(v)^{L_t(v)}\}$. As before, this color $\chi(v)$ is assumed to be a $L_t(v)$-dimensional tuple where each entry takes positive integer values in $[\lambda_t(v)]$. Note that the dimension of tuple and the range of each dimension of the tuple can change with time and we need to be careful about that.
The definitions of  $N^*_i(v)$ and $D^*_i(v)$ remains the same, except that the range of $i$ is now only till $L_t(v)$. 

The invariant that we maintain for every vertex is similar to Invariant~\ref{inv:main} changed appropriately {\em and} we add the condition that at time $t$ each coordinate is in $[\lambda_t(v)]$.
\begin{invariant}
	\label{inv:main-general}
For all $t$, for all $v \in V$,  $i \in [0, L_t(v)]$, we have
(1)  $D^*_i(v) \leq (\deg_t(v)/\lambda_t(v)^i) \cdot f_{t,v}(i)$,
where $f_{t,v}(i) = ( (\lambda_t(v)+1)/(\lambda_t(v) -1))^i$, 
and (2) $\chi^{(t)}_i(v) \in \{1,2,\ldots,\lambda_t(v)\}$.
\end{invariant}

Let us take care of situations when part (2) of the above invariant is violated because the other part is similar to as done in Section~\ref{sec:det:coloring}.
To do so, for a positive integer $p$, define $d_p$ to be the largest value of $\deg_t(v)$ for which $\lambda_t(v)$ evaluates to $p$. 
By the definition of the parameters (since $\lambda_t(v) = \Theta(\lg \deg_t(v))$) we get $d_{p+1} - d_p = \Theta(d_{p})$.
To take care of part (2) of Invariant~\ref{inv:main-general}, the algorithm given in Figure~\ref{fig:while} needs to be changed in Step 4 as follows: if $\deg_t(v) \in [d_p,d_{p+1}]$ (and therefore $\lambda_t(v)=p+1$), we search for $\alpha$ in $[\lambda_t(v)]$ if $\deg_t(v) > d_p + (d_{p+1}-d_p)/2$, otherwise we search for $\alpha$ in $[\lambda_t(v) - 1]$. 

Whenever part (2) of Invariant~\ref{inv:main-general} is violated by vertex $v$ at time $t$, we perform the following changes.
We take $\Theta(\deg_t(v))$ time to find a color for $v$ such that $\chi_i(v) \in [\lambda_t(v)]$ for all $i\in [L_t(v)]$ satisfying Invariant~\ref{inv:main-general}. We now show that this time can be charged to edge deletions {\em incident on $v$} that has happened in the past, and furthermore these edge deletions will not be charged to again.

Firstly note that $v$ violated part (2) on Invariant~\ref{inv:main-general} only because we delete an edge $(u,v)$ and $\deg_t(v)$ and therefore $\lambda_t(v)$ has gone down. 
Note that it must be the case $\chi^{(t)}_i(v) = \lambda_t(v) + 1 = \lambda_{t^-}(v)$, where the third term is the value of $\lambda_t(v)$ just before the edge deletion.
Suppose $\lambda_t(v) = p$. Note that $\deg_t(v)$ must be $d_p$ since at time $\deg_{t^-}(v) = \deg_t(v) + 1$ and $\lambda_{t^-}(v) = p + 1$.
Look at the last time $t'$ before $t$ at which $\chi^{(t')}_i(v)$ was set to $p+1$. At that time, because of the modification made above to the algorithm, we must have had $\deg_{t'}(v) > d_p + (d_{p+1} - d_p)/2$.
Therefore between $t'$ and $t$ we must have had at least $(d_{p+1}-d_p)/2 = \Theta(d_p) = \Theta(\deg_t(v))$ edge deletions incident on $v$. We charge the $\Theta(\deg_t(v))$ time taken to recolor $v$ due to violation of part (2) of invariant to these deletions. Note that since we choose $t'$ to be the last time before $t$, we won't charge to these edge deletions again.

To maintain part (1) of the invariant, the algorithm is similar as in the previous section with two changes: (1) firstly, the WHILE loop checks the new invariant, and (2) is the technical change described above. Observe that even when we delete an edge we run the risk of the invariant getting violated since the RHS of the invariant also goes down. 
For the analysis, the one line argument why everything generalizes is that our analysis is in fact vertex -by-vertex.
More precisely, we have a version of Lemma~\ref{lm:run} where the $\Delta,L,\lambda,f$ are replaced by $\deg_t(v), L_t(v), \lambda_t(v), f_{t,v}$.
Similar changes are in all the other claims and lemmas and for brevity we don't mention the subscripts below.
For instance, time analysis of Lemma~\ref{lm:run} for the WHILE loop taking care of vertex  $x$ generalizes with $L_t(x), \lambda_t(x),f_{t,x}$ and $\deg_t(x)$ replacing $L,\lambda,f$ and $\Delta$. Similarly, in Lemma~\ref{lm:analyze:1} we have exactly the same changes which reflects in Observation~\ref{ob:2}. That is, the decrease in potential in a single while loop is charged to the running time of that while loop. The rest of the analysis as in Section~\ref{sec:det:coloring}.
There is an extra change in Lemma~\ref{lm:analyze:1} where because of the technical change we made, we only get 
$|S^+_x| < (\Delta/\lambda^k)\cdot f(k-1)\cdot \lambda/(\lambda-2)$ since we could be searching over a range of $[\lambda - 1]$. This only increases the update time by an extra factor which is $O(1)$ if $\lambda \geq 2$. 

\subsection{Deterministic $(2\Delta_t-1)$ edge coloring.}
We assert the   invariant that every edge $(u,v)$ gets a color from the palette 
$\{1, \ldots,2\max(\deg_t(u),\deg_t(v))-1\}$. In the subroutine COLOR($e$), the upper bound $u$ is then set to 
$u=\deg_{t}(u)+\deg_{t}(v)$, and the rest of the algorithm remains the same. The trees $T_u$ and $T_v$ now need to be dynamically balanced; but this can be done in $O(\log \Delta_t)$ time using say red-black trees. The other place where the algorithm needs to change is that when an edge $e=(u,v)$ is deleted, the degree of both $u$ and $v$ go down. That may lead to at most four edges (two incident on $u$ and two incident on $v$) violating the invariant. They need to be re-colored using COLOR() procedure again. But this takes $O(\log \Delta_t)$ time in all.

\section{Open Problems}

One obvious open question left from this work is whether we can maintain a $(\Delta+1)$-vertex coloring in polylogarithmic time using a {\em deterministic} algorithm. We believe that this is an important question, since it may help in understanding how to develop deterministic dynamic algorithms in general. It is very challenging and interesting to design deterministic dynamic algorithms with performances similar to the randomized ones for many dynamic graph problems such as maximal matching~\cite{BaswanaGS11,BhattacharyaHI15s,BhattacharyaHI15,BhattacharyaHN16,BhattacharyaCH17,BhattacharyaHN17}, connectivity~\cite{KapronKM13,NanongkaiSW17,Wulff-Nilsen17,NanongkaiS17}, and shortest paths~\cite{Bernstein17,BernsteinC17,BernsteinC16,HenzingerKN14,HenzingerKN16}.  

Another obvious question is whether our deterministic update time for  $(1+o(\Delta))\Delta$-vertex coloring can be improved. We did not try to optimize the polylog factors hidden inside Theorem~\ref{th:analysis} and believe that it can be improved; however, getting an $O(\log \Delta)$ deterministic update time seems challenging. It will also be interesting to get $O(\text{poly} \log \Delta)$ worst-case update time for dynamic vertex coloring.

Finally, one other direction is to study the classes of locally-fixable problems and SLOCAL~\cite{GhaffariKM17}. 
Does every locally-fixable problem admit polylogarithmic update time? How about problems in SLOCAL such as maximal independent set and minimal dominating set?

\section{Acknowledgements} We thank an anonymous reviewer for the proofs of Lemma~\ref{lm:fix:1} and Claim~\ref{cl:invariant:boundary}. 

This project has received funding from the European Research Council (ERC) under the European UnionÕs Horizon 2020 research and innovation programme under grant agreement No 715672. Danupon Nanongkai was also partially supported by the Swedish Research Council (Reg. No. 2015-04659).  The research leading to these results has  received funding from the European Research Council under the European Union's Seventh Framework Programme (FP/2007-2013) / ERC Grant
Agreement no. 340506.


\printbibliography[heading=bibintoc]

\appendix

\section{Locally-Fixable Problems}
\label{sec:appendix}

We consider graph problems in a way similar to \cite{GhaffariKM17}, where there is a set of {\em states} $S_v$ associated with each node $v$, and each node $v$ has to pick a state $s(v) \in S_v$.  
For a {\em locally-fixable} problem, once every node picks its own state, there is a function $f_v$ that determines whether any given node $v \in V$ is {\em valid} or {\em invalid}.
Crucially, the function $f_v$ satisfies the  two properties stated below. The goal in a locally-fixable problem is to assign a state to each node in such a way that all the nodes become valid.
\begin{property}
	\label{prop:app:1}
	The output of the function $f_v$  depends {\em only on}  the states of $v$ and its neighbors. (In other words,  $f_v$ is a constraint that is {\em local to a node $v$} in that it is defined on the states of $v$ and its neighbors.) 
\end{property}
\begin{property}
	\label{prop:app:2}
	Consider any assignment of states to $v$ and its neighbors where $v$ is invalid. Then, without modifying the states of $v$'s neighbors, there is a way to change the state of $v$ that (a) makes $v$ valid, and (b) does not make any erstwhile valid neighbor of $v$ invalid. 
\end{property}

We  present two examples of locally-fixable problems. Note again that all graph problems below are viewed as assigning states to nodes, and their definitions below are standard. Our main task is to define a function $f_v$ for each node $v$ such that Properties~\ref{prop:app:1} and~\ref{prop:app:2} are satisfied. 

\paragraph{$(\Delta+1)$-vertex coloring.} In this problem, the set of states $S_v$ associated with a node $v$ is  the set of $(\Delta+1)$ colors, and a feasible coloring is when every node has different color from its neighbors.  To show that this problem is locally-fixable, consider the function $f_v$ which determines that  $v$ is valid iff it none of its neighbors have the same state as $v$. 
This satisfies Properties~\ref{prop:app:1} and~\ref{prop:app:2}, since a node $v$ has at most $\Delta$ neighbors and there are $(\Delta+1)$ possible states.

\paragraph{$(2\Delta-1)$-edge coloring.} The problem is as follows. Let $\C$ denote the palette of $2\Delta-1$ colors. Let $n = |V|$ denote the number of nodes. We identify these nodes as $V = \{1, \ldots, n\}$. The state $s(v)$ of a node $v$ is an $n$-tuple $s(v) = (s_1(v), \ldots, s_n(v))$, where $s_i(v) \in \C \cup \bot$ for each $i \in [n]$.  The set $S_v$ consists of all such possible $n$-tuples. 
The element $s_u(v)$ is supposed to be the color of edge $(u, v)$, which should be $\bot$ if edge $(u, v)$ does not exist. Thus, the feasible solution is the one where

\smallskip
\noindent (1) for every nodes $u\neq v$, $s_u(v)\neq \bot$ iff there is an edge $(u, v)$ (i.e. each node only assign colors to its incident edges),
 (2) for every edge $(u, v)$, $s_u(v)=s_v(u)$ (i.e.  $u$ and $v$ agree on the color of $(u, v)$), and
 (3) for every edges $(u, v)\neq (u', v)$, $s_u(v)\neq s_{u'}(v)$ (i.e. adjacent edges should have different colors). 

\smallskip
\noindent
To show that this problem is locally-fixable, consider the function $f_v$ which determines that  $v$ is valid iff the following three conditions hold. 

\smallskip
\noindent (1) $s_i(v) = \bot$ for every $i \in [n]$ that is not a neighbor of $v$. (2) $s_i(v) = s_v(i)\neq \bot$ for every neighbor $i \in [n]$ of $v$. 
(3) $s_i(v) \neq s_j(v)$ for every two neighbors $i, j \in [n]$ of $v$ with $i \neq j$.

\smallskip
\noindent
This satisfies Properties~\ref{prop:app:1} and~\ref{prop:app:2}.

\begin{lemma}\label{lem:locally-fixable-dynamic}
In the dynamic setting, we can  update a valid solution to a locally-fixable problem by making at most two changes after an insertion or deletion of an edge. Furthermore, the  changes  occur only at the endpoints of the edge being inserted or deleted. 
\end{lemma}
\begin{proof}
Consider a locally fixable problem on a dynamic graph $G = (V, E)$. Suppose that every node is valid just before the insertion or deletion of an edge $(u, v)$. 
By Property~\ref{prop:app:1}, due to the insertion/deletion of the edge, only the endpoints $u$ and $v$ can potentially become invalid. By Property~\ref{prop:app:2}, we can consider $u$ and $v$ one at a time in any order and make them valid again without making any new node invalid. 
\end{proof}

\paragraph{Locally-fixable problems constitute a subclass of SLOCAL problems.}  Roughly speaking, the complexity class SLOCAL($1$)~\cite{GhaffariKM17} is as follows (the description closely follows \cite{GhaffariKM17}). 
	 In the SLOCAL(1)
	model, nodes are processed in an arbitrary order. When a node $v$ is processed, it can see the current
	state of its neighbors and compute its output as an arbitrary function of
	this. In addition, $v$ can locally store an arbitrary amount of information, which can be read by later
	nodes as part of $v$'s state. 

\begin{lemma}
	Any locally-fixable problem is in SLOCAL(1).
\end{lemma}
\begin{proof}
We initially assign an arbitrary state $s(v) \in S_v$ to every node $v \in V$. Then we scan the nodes in $V$ in any arbitrary order. While considering a node $v$ during the scan, using Property~\ref{prop:app:2} we change the state of $v$ in such a way which ensures that (1) the node $v$ becomes valid and (2) no previously scanned neighbor of $v$ becomes invalid. Thus, at the end of the scan, we are guaranteed that every node is valid.
\end{proof}

\paragraph{The maximal independent set (MIS) problem is not locally fixable.} Recall the standard definition of MIS where the state $s(v)$ of a node $v$ is either $0$ or $1$, and a feasible solution is such that (i) no two adjacent nodes are both $1$, and (ii) no node whose neighbors are all $0$ is in state $0$. We claim that MIS under this standard definition is not locally-fixable; i.e. we cannot define function $f_v$ for every node $v$ such that Properties~\ref{prop:app:1} and~\ref{prop:app:2} are satisfied. 

To see this, consider the following graph instance $G = (V, E)$, where the node-set is defined as $V = \{u, v, u_{1}, \ldots, u_{\Delta-1}, v_1, \ldots, v_{\Delta-1}\}$. The edge-set $E$ is defined as follows. For every $i \in [\Delta-1]$, there are two edges $(u, u_i)$ and $(v, v_i)$.  Now consider the following MIS solution for the graph $G$. We have $s(u) = s(v) = 1$ and $s(x) = 0$ for all $x \in V \setminus \{u, v\}$. Now, if the edge $(u, v)$ gets inserted, we have to change either $s(u)$ or  $s(v)$ to $0$. Without any loss of generality, suppose that we set $s(u) \leftarrow 0$. Then we also need to set $s(u_i) \leftarrow 1$ for every $i \in [\Delta-1]$. In other words, insertion or deletion of an edge can lead to $\Omega(\Delta)$ many changes in a valid solution. By a contrapositive of Lemma~\ref{lem:locally-fixable-dynamic}, MIS is not locally-fixable.
\end{document}